\documentclass[a4paper,11pt]{article}
\usepackage{graphicx,epsfig}
\usepackage{fancyhdr,fancybox,float}
\usepackage[title]{appendix}
\usepackage{indentfirst}
\usepackage{verbatim}
\usepackage[sort&compress, numbers]{natbib}
\usepackage{geometry}
\usepackage{extarrows,chemarrow,xypic}
\usepackage{color}
\usepackage[small]{caption2}
\usepackage{microtype}
\DisableLigatures[f]{encoding = *, family = *}

\usepackage{times}                   

\usepackage{amssymb}                 
\usepackage{amsthm}
\usepackage{mathrsfs}                
\usepackage{bbm}
\usepackage{hyperref}                

\newcommand{\paperfont}{\fontsize{11pt}{1.2\baselineskip}\selectfont}
\pagebreak[4]

\begin{document}
	
\theoremstyle{definition}
\makeatletter
\thm@headfont{\bf}
\makeatother
\newtheorem{theorem}{Theorem}[section]
\newtheorem{definition}[theorem]{Definition}
\newtheorem{lemma}[theorem]{Lemma}
\newtheorem{proposition}[theorem]{Proposition}
\newtheorem{corollary}[theorem]{Corollary}
\newtheorem{remark}[theorem]{Remark}
\newtheorem{example}[theorem]{Example}
\newtheorem{assumption}[theorem]{Assumption}

\lhead{}
\rhead{}
\lfoot{}
\rfoot{}

\renewcommand{\refname}{References}
\renewcommand{\figurename}{Figure}
\renewcommand{\tablename}{Table}
\renewcommand{\proofname}{Proof}
	
\newcommand{\diag}{\mathrm{diag}}
\newcommand{\tr}{\mathrm{tr}}
\newcommand{\re}{\mathrm{Re}}
\newcommand{\one}{\mathbbm{1}}
\newcommand{\Pnum}{\mathbb{P}}
\newcommand{\Enum}{\mathbb{E}}
\newcommand{\Rnum}{\mathbb{R}}
\newcommand{\dnum}{\mathrm{d}}
\newcommand{\hyper}{{}_2F_1}
\newcommand{\confl}{{}_1F_1}

\title{\textbf{Steady-state joint distribution for first-order stochastic reaction kinetics}}
\author{Youming Li$^{1}$,\;\;\;Da-Quan Jiang$^{2,3,*}$,\;\;\;Chen Jia$^{1,*}$ \\
\footnotesize $^1$ Applied and Computational Mathematics Division, Beijing Computational Science Research Center, Beijing 100193, China. \\
\footnotesize $^2$ LMAM, School of Mathematical Sciences, Peking University, Beijing 100871, China. \\
\footnotesize $^3$ Center for Statistical Science, Peking University, Beijing 100871, China. \\
\footnotesize $^*$ Correspondence: jiangdq@math.pku.edu.cn (D.-Q. Jiang), chenjia@csrc.ac.cn (C. Jia)}

\date{}
\maketitle
\thispagestyle{empty}

\paperfont

\begin{abstract}
While the analytical solution for the marginal distribution of a stochastic chemical reaction network has been extensively studied, its joint distribution, i.e. the solution of a high-dimensional chemical master equation, has received much less attention. Here we develop a novel method of computing the exact joint distributions of a wide class of first-order stochastic reaction systems in steady-state conditions. The effectiveness of our method is validated by applying it to four gene expression models of biological significance, including models with 2A peptides, nascent mRNA, gene regulation, translational bursting, and alternative splicing.
\end{abstract}

\section{Introduction}
Stochastic modeling of chemical reaction networks has attracted massive attention in recent years due to its wide applications in biology, chemistry, ecology, and epidemics \cite{anderson2015stochastic}. If a reaction system is well mixed and the number of molecules is very large, random fluctuations can be ignored and the evolution of concentrations of all chemical species can be modeled deterministically as a set of ordinary differential equations (ODEs) based on the law of mass action. If the chemical species are present in low numbers, however, random fluctuations can no longer be ignored and the evolution of copies numbers of all species is usually modeled stochastically as a Markov jump process whose dynamics is governed by the well-known chemical master equation (CME). Thus far, stochastic chemical reaction networks have become a fundamental model for single-molecule enzymology \cite{qian2002single, jia2012kinetic} and single-cell gene expression dynamics \cite{paulsson2005models}. Over the past two decades, the marginal distributions of stochastic reaction systems, such as Michalies-Menten enzyme kinetics \cite{schnoerr2014complex, holehouse2020stochastic}, gene expression dynamics \cite{peccoud1995markovian, shahrezaei2008analytical, zhou2012analytical}, and gene regulatory networks \cite{hornos2005self, grima2012steady, vandecan2013self, kumar2014exact, bokes2015protein, jia2020small, jia2020dynamical}, have been studied extensively by solving the CME exactly or approximately based on various methods. These approaches include the generating function method \cite{peccoud1995markovian}, method of characteristics \cite{shahrezaei2008analytical}, multiscale techniques \cite{melykuti2014equilibrium}, moment closure approximation \cite{lakatos2015multivariate}, moment convergence method \cite{zhang2016moment}, linear noise approximation \cite{thomas2014phenotypic}, linear mapping approximation \cite{cao2018linear}, etc.

The joint distribution of all chemical species for stochastic chemical reaction kinetics has received relatively little attention. Mathematically, the steady-state distribution of a reaction system corresponds to the eigenvector associated with the zero eigenvalue of the rate matrix of the underlying Markovian model. It can always be solved analytically when the rate matrix is finite-dimensional. However, for most reaction systems, the rate matrix is infinite-dimensional since the numbers of reactants are not bounded. In this case, simple approaches like diagonalization of the rate matrix usually fail. Due to the limitation of techniques, the joint distribution can only be solved for some particular systems. It has long been known \cite{krieger1960first, darvey1966stochastic} that (i) the steady-state joint distribution of a closed monomolecular system, which only includes reactions of the form $S_i\rightarrow S_j$, must be a multinomial distribution and (ii) the joint distribution of a detailed balanced reaction network is given by a product of Poissons \cite{van1976equilibrium}. Here detailed balance means that there is no net flux between any pair of reversible reactions. The CME for an open monomolecular system, which consists of synthesis reactions $\varnothing\rightarrow S_i$, degradation reactions $S_i\rightarrow \varnothing$, and conversion reactions $S_i\rightarrow S_j$, has also been solved exactly and the steady-state joint distribution is given by a product-form Poisson distribution \cite{gans1960open, gadgil2005stochastic, heuett2006grand, jahnke2007solving}. Recently, this result has been extended to general stochastic reaction networks that are complex balanced. Here complex balance means that the flux flowing into each complex (see \cite{horn1972general} for definition) is precisely balanced by the flux flowing out of that complex \cite{horn1972general}. In fact, the steady-state joint distribution of a complex balanced reaction network is also given by a product-form Poisson-like distribution \cite{anderson2010product, cappelletti2016product}.

However, the condition of complex balance is very restrictive and not applicable to most systems of biological relevance. If complex balance is not satisfied, the joint distribution has been analytically derived for hierarchic first-order reaction networks \citep{reis2018general}. In the context of stochastic gene expression, the joint distribution for the copy numbers of mRNA and protein has been exactly solved for the two-stage model involving transcription and translation \cite{bokes2012exact, pendar2013exact} and the joint distribution for the copy numbers of two mRNA isoforms has also been analytically derived in the presence of alternative splicing \cite{wang2014alternative}. In addition, the joint distributions of gene expression models have also been studied using the linear noise approximation in the limit of large system size \cite{thomas2014phenotypic}. In most previous papers, the closed-form solution of the joint distribution is computed by first converting the CME into a system of partial differential equations (PDEs) satisfied by the generating function and then solving the system of PDEs using the method of characteristics. However, this method is often very difficult to apply because of the tedious computations involved. Thus far, there is still a lack of a simple and effective approach that can be applied to a wide class of first-order reaction networks.

In this article, we propose a novel and effective method of computing the joint distribution of a first-order reaction system in steady-state conditions. The key idea is to simplify the Markovian model of stochastic reaction kinetics to a modified Markovian model by allowing all zero-order reactions to occur only when all chemical species have zero copies. It turns out that the modified model has a much simpler state space and thus its joint distribution is much easier to solve. Once the modified model is solved analytically, the joint distribution of the original model is automatically obtained by making a simple transformation. Compared with the classical method of characteristics, our approach greatly reduces the theoretical complexity. The paper is organized as follows. In Section 2, we describe the stochastic model of first-order reaction networks and introduce our method in detail. In Section 3, we validate the effectiveness of our approach by applying it to four gene expression models of biological significance. These models include (i) a gene expression model involving 2A self-cleaving peptides, (ii) a multi-step gene expression model involving nascent mRNA, (iii) a gene regulatory model involving translational bursting, and (iv) a multi-step gene expression model involving alternative splicing. We conclude in Section 4.

\section{Model and methods}\label{methods}
A chemical reaction involving a set of chemical species $S_1,\dots,S_N$ can be written in the following general form:
\begin{equation*}
\mu^1S_1+\mu^2S_2+\dots+\mu^NS_N\xrightarrow{k}\nu^1S_1+\nu^2S_2+\dots+\nu^NS_N,
\end{equation*}
where $\mu^i$ and $\nu^i$ are nonnegative integers and $k$ is the rate constant. The order of this reaction is the sum of coefficients of all the reactants, i.e. $\mu^1+\mu^2+\dots+\mu^N$. Following the definition in \citep{anderson2015stochastic}, a reaction system is said to be \emph{first-order} if it only consists of zero-order and first-order reactions. By definition, a first-order reaction system can be written in the following general form:
\begin{gather*}
R_{0j} \colon \varnothing \xrightarrow{k_{0j}} \nu_{0j}^1 S_1+\nu_{0j}^2 S_2+\dots+\nu_{0j}^N S_N,\;\;\; j=1,\dots,r_0,\\
R_{ij} \colon S_i \xrightarrow{k_{ij}} \nu_{ij}^1 S_1+\nu_{ij}^2 S_2+\dots+\nu_{ij}^N S_N,\;\;\; i=1,\dots,N,\;\;\; j=1,\dots,r_i,
\end{gather*}
where $R_{0j}$, $j = 1,\dots,r_0$ are all zero-order reactions involved in the system and $R_{ij}$, $j = 1,\dots,r_i$ are all first-order reactions associated with the reactant $S_i$. For convenience, we write $\nu_{ij}=(\nu^1_{ij},\dots,\nu^N_{ij})$ for each $i = 0,1,\dots,N$ and $j=1,\dots,r_i$. A first-order reaction system can include synthesis reactions $\varnothing \rightarrow S_i$, degradation reactions $S_i\rightarrow \varnothing$, conversion reactions $S_i\rightarrow S_j$, catalytic reactions $S_i\rightarrow S_i+S_j$, and splitting reactions $S_i\rightarrow S_j+S_k$; hence it can be widely applied to model various naturally occurring systems in biology and physics.

We next focus on the stochastic dynamics of a first-order reaction network. The microstate of the system can be described by an ordered $N$-tuple $n = (n_1,\dots,n_N)$, where $n_i$ denotes the molecule number of $S_i$. Based on the law of mass action, the stochastic dynamics of the system can be described by a Markov jump process whose transition rates are given by
\begin{equation}\label{rateoriginal}
\begin{gathered}
q_{n,n+\nu_{0j}} = k_{0j}, \;\;\; 1\leq j\leq r_0,\\
q_{n,n+\nu_{ij}-e_i} = k_{ij}n_i, \;\;\; 1\leq i\leq N, \;\;\; 1\leq j\leq r_i,
\end{gathered}
\end{equation}
where $q_{n,n'}$ denotes the transition rate from microstate $n$ to microstate $n'$, $\nu_{0j}$ is the reaction vector of the zero-order reaction $R_{0j}$, i.e. the vector indicating the species change after the reaction, and $\nu_{ij}-e_i$ is the reaction vector of the first-order reaction $R_{ij}$ with $e_i$ being the vector whose $i$th component is $1$ and all other components are zero.

Throughout this paper, we assume that the reaction system is ergodic, which guarantees that the system has a unique steady-state distribution. Let $p_n = p_{n_1,\cdots,n_N}$ denote the probability of observing microstate $n$. Then the evolution of the Markovian system is governed by the CME
\begin{equation}
\dot{p}_n =
\sum_{i=1}^N\sum_{j=1}^{r_i}k_{ij}[(n_i+1-\nu^i_{ij})p_{n+e_i-\nu_{ij}}-n_ip_n]+\sum_{j=1}^{r_0}k_{0j}[p_{n-\nu_{0j}}-p_n],
\end{equation}
where the first term on the right-hand side corresponds to the occurrence of first-order reactions and the second term corresponds to the occurrence of zero-order reactions. To proceed, let
\begin{equation*}
F(x_1,\cdots,x_n) = \sum_{n_1,\cdots,n_N}p_{n_1,\dots,n_N} x_1^{n_1}\dots x_N^{n_N}
\end{equation*}
denote the generating function associated with the joint distribution $p_{n_1,\dots,n_N}$. Then $F$ satisfies the following PDE \cite{reis2018general}:
\begin{equation}\label{f}
\frac{\partial F}{\partial t}=\sum_{i=1}^N\sum_{j=1}^{r_i} k_{ij}\left(x^{\nu_{ij}}-x_i\right)\frac{\partial F}{\partial x_i}+\sum_{j=1}^{r_0} k_{0j}\left(x^{\nu_{0j}}-1\right)F.
\end{equation}
A classical method of solving the CME is to first solve Eq. \eqref{f} to obtain the closed form of the generating function $F$ and then recover the joint distribution $p_n$ by taking the derivatives of $F$ at zero. However, it is remarkably difficult to solve Eq. \eqref{f} analytically in most cases, even at the steady state.

Here we propose a novel method of solving the CME in steady-state conditions. To this end, we construct a simpler Markov jump process called the \emph{modified Markovian model}. The microstate of the modified model is still described by an ordered $N$-tuple $n = (n_1,n_2,\dots,n_N)$. Note that for the original model, the zero-order reaction $R_{0j}$ can lead to a transition from any microstate $n$ to microstate $n+\nu_{0j}$; in other words, the zero-order reactions can occur at any microstate of the original model. However, for the modified model, we only allow the zero-order reactions to occur at the microstate $\mathbf{0}=(0,\dots,0)$, which is called the \emph{zero microstate}, while the first-order reactions follow the same transition rule as the original model. To summarize, the transition rates for the modified model are given as follows:
\begin{equation}\label{ratemodified}
\begin{gathered}
\tilde{q}_{n,n+\nu_{0j}} = \begin{cases} 0, & n\neq\mathbf{0},\\
k_{0j}, & n=\mathbf{0},
\end{cases} \;\;\; 1\leq j\leq r_0,\\
\tilde{q}_{n,n+\nu_{ij}-e_i} = k_{ij}n_i, \;\;\; 1\leq i\leq N,\;1\leq j\leq r_i.
\end{gathered}
\end{equation}
Comparing Eqs. \eqref{rateoriginal} and \eqref{ratemodified}, we can see that the modified model can be easily derived from the original one by eliminating those transitions from $n$ to $n+\nu_{0j}$ for $n\neq\mathbf{0}$. Let $\pi_n = \pi_{n_1,n_2,\dots,n_N}$ denote the probability of observing microstate $n$ for the modified model and let
\begin{equation*}
H(x_1,\dots,x_N)=\sum_{n}\pi_{n_1,\dots,n_N}x_1^{n_1}\dots x_N^{n_N}
\end{equation*}
denote its generating function. Then the evolution of the modified model is then governed by the master equation
\begin{equation}\label{mastermodified}
\begin{split}
\dot{\pi}_n =&\; \sum_{i=1}^N\sum_{j=1}^{r_i}\tilde{q}_{n+e_i-\nu_{ij},n}\pi_{n+e_i-\nu_{ij}}
-\sum_{i=1}^N\sum_{j=1}^{r_i}\tilde{q}_{n,n+\nu_{ij}-e_i}\pi_n\\
&\;+\sum_{j=1}^{r_0}\tilde{q}_{n-\nu_{0j},n}\pi_{n-\nu_{0j}}
-\sum_{j=1}^{r_0}\tilde{q}_{n,n+\nu_{0j}}\pi_{n}.
\end{split}
\end{equation}
where the first two terms on the right-hand side correspond to first-order reactions and the last two terms correspond to zero-order reactions. We next make a crucial observation that if the system contains at least one zero-order reaction and both the original and modified models have reached the steady state, then the two generating functions $F$ and $H$ are related by (see Appendix A for the proof)
\begin{equation}\label{expression}
F(x_1,\dots,x_N) = e^{\frac{H(x_1,\dots,x_N)-1}{\pi_{\mathbf{0}}}},
\end{equation}
where $\pi_\mathbf{0}$ is the probability of observing the zero microstate for the modified model. In general, the modified model has a simpler transition diagram than the original model and thus the master equation for the former is much easier to solve. Once we have obtained the generating function $H$ of the modified model, we can use Eq. \eqref{expression} to compute the generating function $F$ of the original model. Finally, the steady-state joint distribution for the copy numbers of all chemical species can be recovered by taking the derivatives of $F$ at zero, i.e.
\begin{equation*}
p_n = \frac{1}{n_1!\cdots n_N!}\frac{\partial^{n_1+\dots+n_N}F}{\partial x_1^{n_1}\dots\partial x_N^{n_N}}(0,\dots,0).
\end{equation*}

We summarize the above method as follows: first, we construct the modified model (which is usually much simpler than the original model) and compute its steady-state joint distribution $\pi_n$; next, we calculate the generating function $H$ of the modified model and use Eq. \eqref{expression} to compute the generating function $F$ of the original model; finally, we recover the steady-state joint distribution of the original model by taking the derivatives of $F$. We emphasize that Eq. \eqref{expression} does not hold if the two models have not reached the steady state. In fact, the proof of Eq. \eqref{expression} relies on the close relationship between the partial derivatives of $F$ and $H$ with respect to $x_i$ in steady-state conditions, while for the time-dependent case, we need to take the partial derivatives with respect to $t$ into consideration, which invalidates our approach (see Appendix A for details).
\begin{figure}[!htb]
\centering\includegraphics[width=120mm]{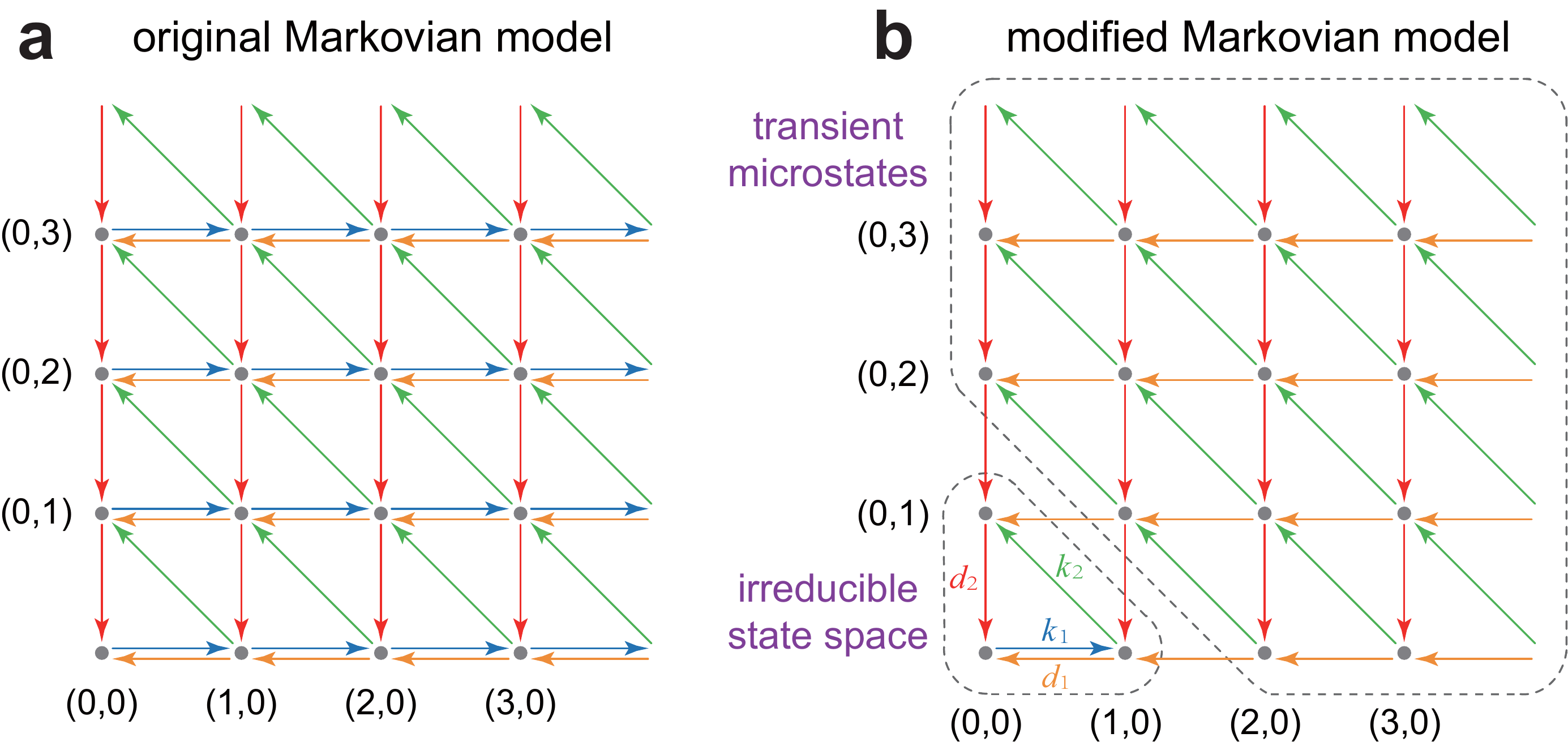}
\caption{\textbf{Transition diagrams of the original and modified Markovian models for the reaction scheme given in Eq. \eqref{example1}.} The red arrows correspond to the reaction $P_2\rightarrow\varnothing$, the green arrows correspond to $P_1\rightarrow P_2$, the blue arrows correspond to $\varnothing\rightarrow P_1$, and the orange arrows correspond to $P_1\rightarrow \varnothing$. For the modified model, we only allow the zero-order reaction (blue arrows) to occur at the zero microstate. The irreducible state space of the modified model only consists of three microstates: $(0,0)$, $(1,0)$, and $(0,1)$.}\label{comparison}
\end{figure}

We next focus on the transition diagrams of the two models. Recall that a microstate $n$ is called recurrent if there exists a path in the transition diagram that starts from $n$ and returns to itself; otherwise it is called transient. Actually, transient microstates contribute nothing to the steady-state probabilities and thus the steady-state distribution is only concentrated on the collection of all recurrent microstates, which is called the irreducible state space \cite{norris1998markov, jia2016model}. Hence for both models, we only need to focus on the irreducible state space, instead of the whole state space. We now use a simple example to show the relationship between the two models. Consider the following open monomolecular system:
\begin{equation}\label{example1}
\varnothing \autorightleftharpoons{$k_1$}{$d_1$}P_1 \xrightarrow{k_2}P_2 \xrightarrow{d_2}\varnothing.
\end{equation}
If we regard $P_1$ and $P_2$ as two conformational states of a protein, then this reaction scheme describes the synthesis, degradation, and conformational changes for the protein. The transition diagrams for the original and modified models associated with this reaction scheme are depicted in Fig. \ref{comparison}. The difference between them is that the zero-order reaction $\varnothing\rightarrow P_1$ (blue arrows) can occur at any microstate for the original model, but it can only occur at only the zero microstate for the modified model. It can be seen from Fig. \ref{comparison}(b) that the modified model has many transient microstates. The irreducible state space (the collection of all recurrent microstates) for the original model is the whole two-dimensional nonnegative integer lattice, while the irreducible state space for the modified model is simply the collection of the following three microstates:
\begin{equation*}
\{(0,0), (1,0), (0,1)\},
\end{equation*}
which is much simpler than that of the original model. Once the modified model enters the irreducible state space, it can never leave it anymore. Since the steady-state distribution of the modified model is only concentrated on the irreducible state space which contains only three microstates, we immediately obtain
\begin{equation*}
\pi_{0,0}=\frac{1}{1+a+b},\;\;\;
\pi_{1,0}=\frac{a}{1+a+b},\;\;\;
\pi_{0,1}=\frac{b}{1+a+b},
\end{equation*}
where $a = k_1/(k_2+d_1)$ and $b = k_1k_2/(k_2+d_1)d_2$. Thus the generating function of the modified model is given by $H(x_1,x_2) = \pi_{0,0}+\pi_{1,0}x_1+\pi_{0,1}x_2$. It then follows from Eq. \eqref{expression} that the generating function of the original model is given by
\begin{equation*}
F(x_1,x_2) = e^{\frac{\pi_{0,0}+\pi_{1,0}x_1+\pi_{0,1}x_2-1}{\pi_{0,0}}}
= e^{a(x_1-1)+b(x_2-1)}.
\end{equation*}
Then the steady-state joint distribution for the copy numbers of $P_1$ and $P_2$ can be recovered by taking the derivatives of $F$ at zero, which finally gives
\begin{equation*}
p_{n_1,n_2} = \frac{a^{n_1}b^{n_2}}{n_1!n_2!}e^{-(a+b)}.
\end{equation*}	
Note that this is the product of two Poisson distributions. In fact, it has been shown in \cite{gadgil2005stochastic} that the joint distribution of an open monomolecular system must be a product of Poissons, which is consistent with our result. However, compared with the derivation in \cite{gadgil2005stochastic}, our method is much simpler.

Our method can also be used to compute many other quantities of interest. First, the steady-state marginal distribution for the copy number of any chemical species can be easily computed. To see this, let $p^i_{n_i}$ denote the steady-state probability of having $n_i$ copies of $S_i$. Then the marginal distribution can be recovered from the generating function $F$ as
\begin{equation}\label{marginaldist}
p^i_{n_i}=\frac{1}{n_i!}\frac{\partial^{n_i}F}{\partial x_i^{n_i}}(1,\cdots,0,\cdots,1),
\end{equation}
where $(1,\cdots,0,\cdots,1)$ is the vector whose $i$th component is $0$ and other components are all $1$. Note that the generating function given in Eq. \eqref{expression} is a composite function. The following Fa$\grave{\text{a}}$ di Bruno's formula \cite{johnson2002curious} gives the explicit expression for the higher-order derivatives of a composite function:
\begin{equation*}
\frac{d^n}{d x^n}f(g(x))=\sum_{k=1}^n f^{(k)}(g(x))B_{n,k}(g'(x),g''(x),\dots,g^{(n-k+1)}(x)),
\end{equation*}
where $B_{n,k}(x_1,\dots,x_{n-k+1})$ is the incomplete Bell polynomial \cite{bell1927partition}. The above two equations, together with Eq. \eqref{expression}, give the following analytical expression for the marginal distributions of all species:
\begin{equation}\label{marginal}
p^i_{n_i} = \frac{B_{n_i}(g_{i,1},g_{i,2},\cdots,g_{i,n_i})}{n_i!}e^{\frac{H(1,\cdots,0,\cdots,1)-1}{\pi_{\mathbf{0}}}},
\end{equation}
where
\begin{equation*}
B_n(x_1,\cdots,x_n) = \sum_{k=1}^nB_{n,k}(x_1,\cdots,x_{n-k+1})
\end{equation*}
is the complete Bell polynomial \cite{bell1927partition}, and
\begin{equation*}
g_{i,k} = \frac{1}{\pi_{\mathbf{0}}}\frac{\partial^k H}{\partial x_i^k}(1,\cdots,0,\cdots,1),\;\;\;
k = 1,\dots,n_i.
\end{equation*}
In addition, the steady-state mean and variance for the copy number of $S_i$ can be obtained as
\begin{gather*}
\langle n_i\rangle = \frac{\partial F}{\partial x_i}(1,\dots,1),\label{mean}\\
\sigma^2_{n_i} = \left[\frac{\partial^2 F}{\partial x^2_i}+\frac{\partial F}{\partial x_i}-\left(\frac{\partial F}{\partial x_i}\right)^2\right](1,\dots,1),\label{variance}
\end{gather*}
where $\sigma^2_{n_i} = \langle n_i^2\rangle-\langle n_i\rangle^2$ denotes the copy number variance of $S_i$. Finally, the steady-state covariance for the copy numbers of any pair of chemical species $S_i$ and $S_j$ can be computed as
\begin{equation*}
\mathrm{Cov}(n_i,n_j) = \langle n_in_j\rangle -\langle n_i\rangle \langle n_j\rangle
= \left[\frac{\partial^2 F}{\partial x_i\partial x_j}-\frac{\partial F}{\partial x_i}\frac{\partial F}{\partial x_j}\right](1,\dots,1).
\end{equation*}
In particular, the correlation coefficient between the copy numbers of $S_i$ and $S_j$ is given by
\begin{equation}\label{correlation}
\rho_{n_i,n_j}=\frac{\mathrm{Cov}(n_i,n_j)}{\sigma_{n_i}\sigma_{n_j}}.
\end{equation}
These formulas will be used to analyze the dynamic properties of some important gene expression models in what follows.

We have seen from the previous example that our method is particularly effective when the modified model has a finite irreducible state space. A natural question is when this occurs. To answer this, we recall that a family of reactions
\begin{equation*}
R_i\colon \mu_i^1S_1+\dots+\mu_i^NS_N\xrightarrow{k_i} \nu^1_{i} S_1+\dots+\nu^N_{i} S_N,
\;\;\;i=1,\dots,r
\end{equation*} has a conservation law, if there exists a nonzero vector $\omega=(\omega_1,\dots,\omega_N)$ such that
\begin{equation*}
\omega_1\mu^1_i+\omega_2\mu^2_i+\dots+\omega_N\mu^N_i = \omega_1\nu^1_i+\omega_2\nu^2_i+\dots+\omega_N\nu^N_i
\end{equation*}
for all $i=1,\dots,r$. In Appendix B, we prove that if all the first-order reactions except degradation reactions have a conservation law with positive coefficients $\omega_1,\cdots,\omega_N> 0$, then the modified model must have a finite irreducible state space. To verify this criterion, we apply it to the reaction scheme given in Eq. \eqref{example1}. For this reaction system, there are three first-order reactions:
\begin{equation*}
P_1\rightarrow P_2,\;P_1\rightarrow\varnothing,\;P_2\rightarrow\varnothing.
\end{equation*}
Among these reactions, only $P_1\rightarrow P_2$ is not a degradation reaction and obviously, it has a conservation law with positive coefficients $\omega_1 = \omega_2 = 1$ since the total number of $P_1$ and $P_2$ is invariant. It then follows from the above criterion that the corresponding modified model has a finite irreducible state space, which is consistent with the previous discussion. Before leaving this section, we emphasize that if a family of reactions contains a first-order catalytic reaction such as $S_i \rightarrow S_i+S_j$, which appears in many biochemical systems, then the family of reactions can never have a conservation law with positive coefficients. In this case, the modified model may have an infinite irreducible state space. Fortunately, for many biochemical systems involving catalytic reactions, our method is still applicable, although the computation will be more complicated than the case of finite irreducible state space. In the next section, we shall apply our method to compute the steady-state joint distributions of mRNAs and/or proteins in four gene expression models of biological significance.

\section{Applications}

\subsection{Gene expression model with 2A self-cleaving peptides}
As the first application, we consider a gene expression system involving 2A self-cleaving peptides, also called 2A peptides. Biologically, 2A peptides are 18-22 amino-acid-long oligopeptides derived from a wide range of viral families \cite{1991Cleavage,Andrea2005Development} that mediate cleavage of polypeptides during translation in eukaryotic cells \cite{2017Systematic}, and therefore enable the synthesis of several gene products (proteins) from a single transcript. For this reason, 2A peptides are widely used in genetic engineering to cleave a long peptide into two shorter peptides. Specifically, the coding region of a 2A peptide (2A) is inserted between the coding regions of two proteins (Fig. \ref{splitf}(a)). The mechanism of 2A-mediated self-cleavage was recently discovered to be ribosome skipping the formation of a peptide bond at the C-terminus of the 2A \cite{2001Analysis,2001The}. There are two possibilities for a 2A-mediated skipping event: (i) successful skipping and recommencement of translation results in two cleaved proteins: the protein upstream of the 2A is attached to the complete 2A peptide except for the C-terminal proline, while the protein downstream of the 2A is attached to one proline at the N-terminus; (ii) successful skipping but ribosome fall-off and discontinued translation results in only the protein upstream of the 2A \cite{2017Systematic}. Then the effective reactions describing the gene expression system are given by
\begin{gather}\label{example2}
G\xrightarrow{k_1}G+P_1+P_2,\;\;\; G\xrightarrow{k_2}G+P_1,\;\;\;
P_1\xrightarrow{d_1}\varnothing,\;\;\;P_2\xrightarrow{d_2}\varnothing,
\end{gather}
where $G$ is the coding region illustrated in Fig. \ref{splitf}(a) and $P_1$ and $P_2$ are two proteins. The first reaction describes ribosome skipping, the second reaction describes ribosome fall-off, and the remaining two reactions describe the degradation of the two proteins. The microstate of the system can be represented by an ordered pair $(n_1,n_2)$, where $n_i$ denotes the copy number of $P_i$. Let $p_{n_1,n_2}$ denote the probability of observing microstate $(n_1,n_2)$ and let
\begin{equation*}
F(x_1,x_2)=\sum_{n_1,n_2}p_{n_1,n_2}x_1^{n_1}x_2^{n_2}
\end{equation*}
denote the corresponding generating function. Then the stochastic gene expression dynamics can be described by a Markov jump process with transition diagram illustrated in Fig. \ref{splitf}(b). The evolution of the Markovian system is governed by the CME
\begin{equation*}
\begin{aligned}
\dot{p}_{n_1,n_2}&=k_1p_{n_1-1,n_2-1}+k_2p_{n_1-1,n_2}+d_1(n_1+1)p_{n_1+1,n_2}\\
& \;\;\;\; +d_2(n_2+1)p_{n_1,n_2+1}-(k_1+k_2+d_1n_1+d_2n_2)p_{n_1,n_2}.
\end{aligned}
\end{equation*}

\begin{figure}[!htb]
\centering\includegraphics[width=1.0\textwidth]{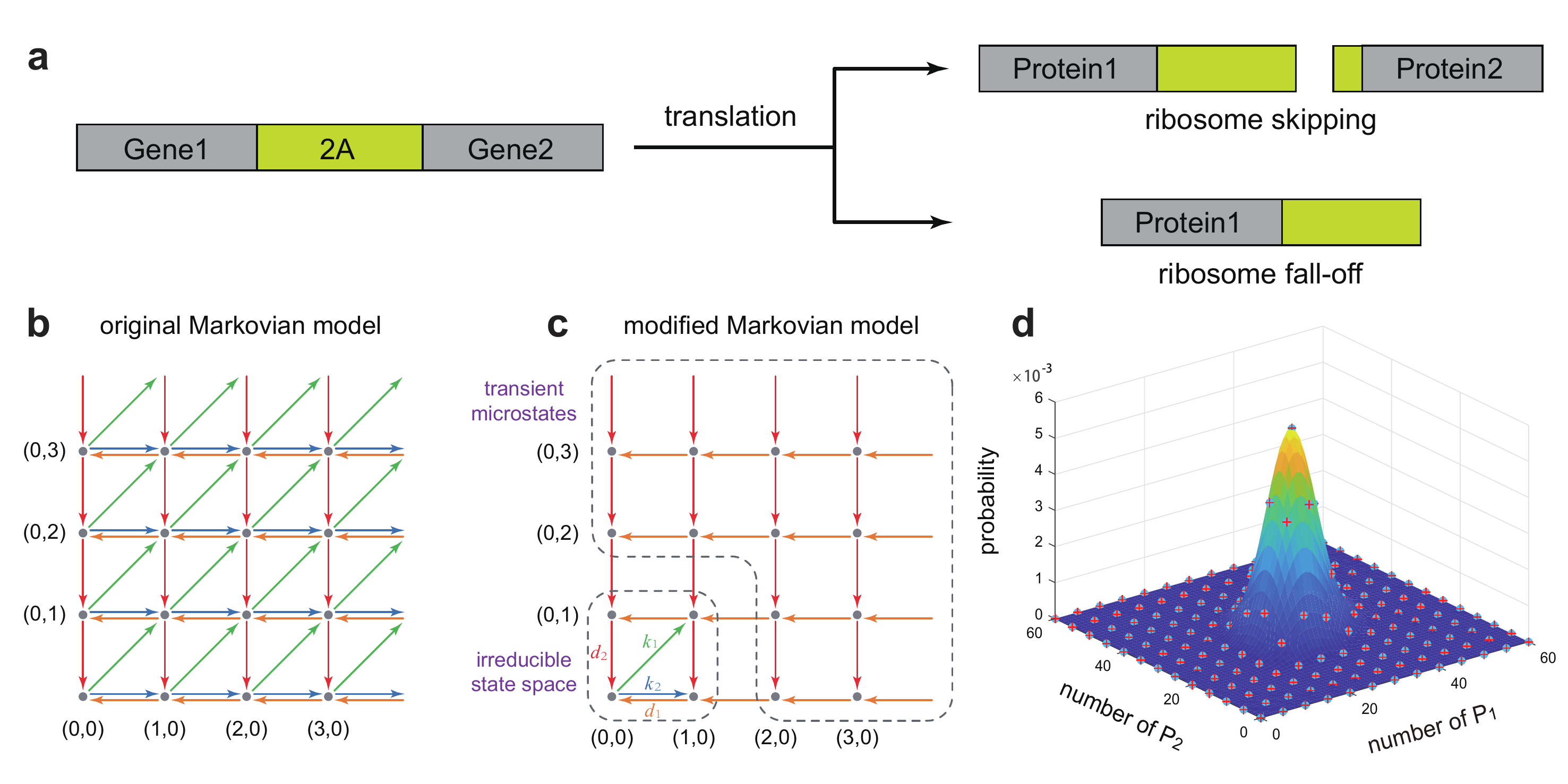}	
\caption{\textbf{A gene expression model involving 2A peptides}. (a) Translation mechanism of two genes with the coding region of a 2A peptide (2A) inserted in between. There are two possibilities: ribosome skipping results in two cleaved proteins and ribosome fall-off results in only the protein upstream of the 2A \cite{2017Systematic}. (b),(c) Transition diagrams for the original and modified models. The green arrows correspond to the reaction $G\rightarrow G+P_1+P_2$, the blue arrows correspond to $G\rightarrow G+P_1$, the red arrows correspond to $P_2\rightarrow \varnothing$, and the orange arrows correspond to $P_1\rightarrow\varnothing$. (d) Comparison of the analytical steady-state joint distribution for the numbers of the two proteins given in Eq. \eqref{distribution} (colored surface) with the numerical simulations obtained using FSP (red plus signs) and stochastic simulations obtained using SSA (light blue dots). Here SSA is performed by generating $80000$ stochastic trajectories. The model parameters are chosen as $k_1=30, k_2=30, d_1=2, d_2=1$.}\label{splitf}
\end{figure}

To solve this CME, we consider the modified Markovian model. We emphasize here that we do not take copy number variation of the gene into account and hence the first two reactions in Eq. \eqref{example2} can be regarded as zero-order reactions. The reason why we explicitly write out the gene $G$, instead of using $\varnothing$, in the first two reactions is to stress that proteins are produced from genes. Since the zero-order reactions can only occur at the zero microstate, the modified model has the transition diagram illustrated in Fig. \ref{splitf}(c). While the transition diagram of the modified model is complicated, the irreducible state space is actually finite and only contains the following four microstates:
\begin{equation*}
\{(0,0),(1,1),(1,0),(0,1)\}.
\end{equation*}
Since the steady-state distribution of the modified model is concentrated on the irreducible state space which contains only four microstates, it can be easily computed as
\begin{equation*}
\pi_{0,0}=\frac{\alpha_0}{\alpha},\;\;\;\pi_{1,0}=\frac{\alpha_1}{\alpha},\;\;\; \pi_{0,1}=\frac{\alpha_2}{\alpha},\;\;\;\pi_{1,1}=\frac{\alpha_{12}}{\alpha},
\end{equation*}
where
\begin{equation*}
\alpha_0 = 1,\;\;\;\alpha_1 = \frac{k_2(d_1+d_2)+k_1d_2}{d_1(d_1+d_2)},\;\;\;
\alpha_2 = \frac{k_1d_1}{d_2(d_1+d_2)},\;\;\;\alpha_{12}=\frac{k_1}{d_1+d_2},
\end{equation*}
and $\alpha = \alpha_0+\alpha_1+\alpha_2+\alpha_{12}$. Then the generating function of the modified model is given by
\begin{equation*}
\begin{aligned}
H(x_1,x_2)&=\frac{1}{\alpha}\left(\alpha_0+\alpha_1x_1+\alpha_{2}x_2+\alpha_{12}x_1x_2\right).\\
\end{aligned}
\end{equation*}
It then follows from Eq. \eqref{expression} that the generating function of the original model is given by
\begin{equation}\label{splittingexample}
F(x_1,x_2) = e^{\frac{H(x_1,x_2)-1}{\pi_{0,0}}}
= e^{\alpha_1(x_1-1)+\alpha_2(x_2-1)+\alpha_{12}(x_1x_2-1)}.
\end{equation}
This shows that the copy numbers of the two proteins have a bivariate Poisson distribution \cite{loukas1986index}, which can be recovered from $F$ by taking the derivatives:
\begin{equation}\label{distribution}
p_{n_1,n_2}
= \frac{1}{n_1!n_2!}\frac{\partial^{n_1+n_2}F}{\partial x_1^{n_1}\partial x_2^{n_2}}\left(0,0\right)
= \sum_{i=0}^{n_1\wedge n_2}\frac{\alpha_1^{n_1-i}\alpha_2^{n_2-i}\alpha^i_{12}}{i!(n_1-i)!(n_2-i)!}
e^{-\left(\alpha_1+\alpha_{2}+\alpha_{12}\right)},
\end{equation}
where $n_1\wedge n_2$ denotes the smaller one of $n_1$ and $n_2$. Taking $x_2=1$ and $x_1=1$ in Eq. \eqref{splittingexample}, we obtain
\begin{equation*}
F(x_1,1) = e^{(\alpha_1+\alpha_{12})(x_1-1)},\;\;\;F(1,x_2) = e^{(\alpha_2+\alpha_{12})(x_2-1)}.
\end{equation*}
It then follows from Eq. \eqref{marginaldist} that the steady-state marginal distributions for the two proteins are given by
\begin{equation}\label{poisson}
p^1_{n_1} = \frac{(\alpha_1+\alpha_{12})^{n_1}}{n_1!}e^{-(\alpha_1+\alpha_{12})},\;\;\;
p^2_{n_2} = \frac{(\alpha_2+\alpha_{12})^{n_2}}{n_2!}e^{-(\alpha_2+\alpha_{12})}.
\end{equation}
This shows that both proteins have a marginal Poisson distribution but their joint distribution is not the product of two Poisson distributions (note that a bivariate Poisson distribution may not be the product of two Poissons). This reaction system should be compared with complex balanced networks. In fact, it was shown in \cite{anderson2010product} that if a reaction system is complex balanced, then the copy numbers of all chemical species must have a product-form Poisson distribution in steady-state conditions, which is very different from the non-product-form Poisson distribution studied here.

To validate our analytical solution, we compare it with the numerical solutions obtained using the finite state projection algorithm (FSP) \cite{munsky2006finite} and the stochastic simulation algorithm (SSA), as illustrated in Fig. \ref{splitf}(d). When using FSP, we truncate the state space at large enough $N_1$ and $N_2$, with $N_1$ and $N_2$ being the truncation sizes for $n_1$ and $n_2$, respectively, and then solve the normalized eigenvector of the truncated rate matrix corresponding to the zero eigenvalue numerically using MATLAB. The truncation sizes are chosen to be $N_1 = 5(k_1+k_2)/d_1$ and $N_2 = 5k_1/d_2$. Since $(k_1+k_2)/d_1$ and $k_1/d_2$ are the typical copy numbers for proteins $P_1$ and $P_2$, respectively, the probability that the protein numbers are outside the truncation region is very small and practically can always be ignored. It can be seen that the analytical solution coincides perfectly with both FSP and SSA. Our analytical results can also be used to analyze the correlation between the two proteins. It follows from Eqs. \eqref{correlation} and \eqref{splittingexample} that the correlation coefficient between the numbers of $P_1$ and $P_2$ is given by
\begin{equation*}
\rho_{P_1,P_2}
= \frac{\alpha_{12}}{\sqrt{\left(\alpha_1+\alpha_{12}\right)\left(\alpha_2+\alpha_{12}\right)}}
= \frac{1}{{\sqrt{\left(1+\frac{k_2}{k_1}\right)\left(1+\frac{d_1}{d_2}\right)
\left(1+\frac{d_2}{d_1}\right)}}}.
\end{equation*}
Clearly, the numbers of the two proteins are always positively correlated and their correlation coefficient has the upper bound
\begin{equation*}
\rho_{P_1,P_2} \leq \frac{1}{2\sqrt{1+\frac{k_2}{k_1}}},
\end{equation*}
where the equality holds if and only if $d_1 = d_2$. This means that the correlation is the strongest when the degradation rates of the two proteins are equal. In addition, we can see that the correlation coefficient is always smaller than $0.5$ and is comparatively large when the degradation rates of the two proteins are close to each other, i.e. $d_1 \approx d_2$, and when the translation rate due to ribosome skipping is much larger compared to the translation rate due to ribosome fall-off, i.e. $k_1\gg k_2$. Before leaving this section, we point out that biologically, it is possible to generate three or more cleaved proteins from a single transcript using coding sequences of multiple 2A peptides \cite{2017Systematic}. In this case, our method can still be used to compute the joint copy number distributions for these proteins since the irreducible state space of the modified model is always finite.

\subsection{Gene expression model with nascent mRNA}\label{presection}
Based on the central dogma of molecular biology, the gene expression dynamics in an individual cell has a standard two-stage representation involving transcription and translation \cite{shahrezaei2008analytical}. In the literature, the transcription step is usually modeled as the elementary reaction $G\rightarrow G+M$, where $G$ is the gene of interest and $M$ is the corresponding mRNA. However, in living cells, the realistic transcription process is much more complicated: first the gene is transcribed to produce the so-called nascent mRNA and then several steps such as 5' capping, 3' polyadenylylation, and mRNA splicing to remove the introns are necessary for the nascent mRNA to become the mature mRNA \cite{saitou2013introduction}. Only the mature mRNA can undergo translation to produce the protein. Recent studies about single-cell RNA-sequencing data analysis have highlighted the need to incorporate the nascent mRNA dynamics into the model in order to introduce the key concept of RNA velocity \cite{la2018rna, li2020mathematics}.

Here we consider a more realistic gene expression model depicted in Fig. \ref{mrnaall}(a). Let $G$ denote the gene of interest, let $M_\star$ denote the nascent mRNA, let $M$ denote the mature mRNA, and let $P$ denote the protein. Then the effective reactions for the gene expression model are given by:
\begin{gather*}
G\xrightarrow{s} G+M_\star,\;\;\; M_\star\xrightarrow{k} M,\;\;\; M\xrightarrow{u}M+P, \\
M_\star\xrightarrow{f}\varnothing,\;\;\; M\xrightarrow{v} \varnothing,\;\;\; P \xrightarrow{d}\varnothing,
\end{gather*}
where the first reaction represents transcription, the second reaction represents the conversion of nascent mRNA into mature mRNA, the third reaction represents translation, and the remaining three reactions represent the degradation of all gene products. If the dynamics of nascent mRNA is ignored, then the steady-state joint distribution of mRNA and protein numbers has been derived in \cite{bokes2012exact}. Here we consider a more complicated model involving nascent mRNA (a similar model has been solved in \cite{pendar2013exact} using a different method). The microstate of the gene can be represented by an ordered triple $(m_\star,m,n)$: the copy number $m_*$ of nascent mRNA, the copy number $m$ of mature mRNA, and the copy number $n$ of protein. Let $p_{m_\star,m,n}$ denote the probability of observing microstate $(m_\star,m,n)$ and let
\begin{equation*}
F(x_\star,x,y)=\sum_{m_\star,m,n}p_{m_\star,m,n}x_\star^{m_\star}x^m y^n
\end{equation*}
denote the corresponding generating function. Then the evolution of the gene expression model is governed by the CME
\begin{equation*}
\begin{aligned}
\dot{p}_{m_\star,m,n}=&\; sp_{m_\star-1,m,n}+k(m_\star+1)p_{m_\star+1,m-1,n}+ump_{m_\star,m,n-1}\\
&\; +f(m_\star+1)p_{m_\star+1,m,n}+v(m+1)p_{m_\star,m+1,n}+d(n+1)p_{m_\star,m,n+1}\\
&\; -(s+km_\star+fm_\star+um+vm+dn)p_{m_\star,m,n}.
\end{aligned}
\end{equation*}

\begin{figure}[!htb]
\centering\includegraphics[width=1.0\textwidth]{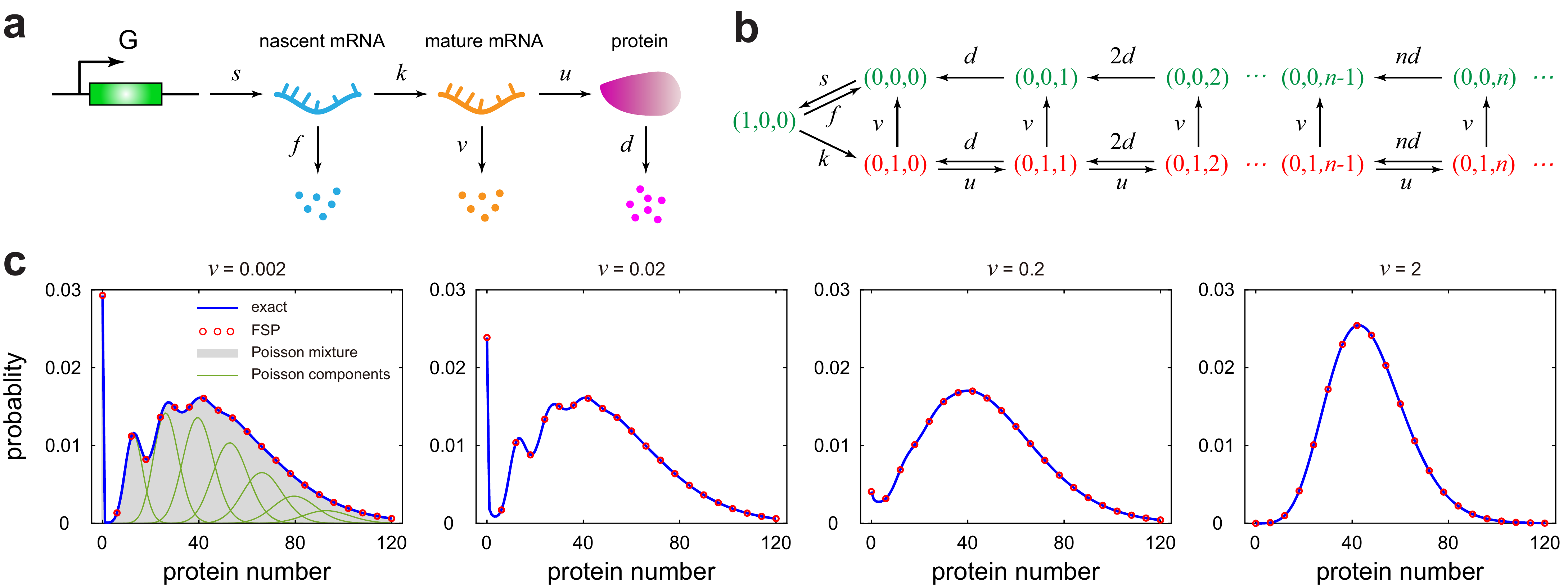}
\caption{\textbf{A multi-step gene expression model}. (a) Schematic of a multi-step gene expression model, which includes transcription, translation, and the production of mature mRNA from nascent mRNA. (b) Transition diagram for the modified model restricted to the irreducible state space. (c) Comparison between the exact steady-state distribution for the protein number given in Eq. \eqref{expression1} (blue curve) and FSP simulations (red circles) as $\nu = v/d$ varies while keeping $b = ks/(k+f)v$ as constant. The left panel also compares the exact distribution (blue curve) with the mixed Poisson approximation given in Eq. \eqref{mixture2} (grey region). The Poisson components of the mixed distribution are shown by the green curves. The model parameters are chosen as $u=40/3, d=1, k=0.002, s=10$ and the parameter $f$ is chosen so that $b=3.52$.}\label{mrnaall}
\end{figure}

To solve this CME, we consider the modified Markovian model. We emphasize again that we do not take copy number variation of the gene into consideration and thus the reaction $G\rightarrow G+M_\star$ can be viewed as a zero-order reaction. Since the zero-order reaction can only occur at the zero microstate, it is easy to see that the irreducible state space of the modified model is given by
\begin{equation*}
\{(1,0,0)\}\cup\{(0,0,n),(0,1,n):\;n\geq 0\},
\end{equation*}
and the transition diagram restricted to the irreducible state space is illustrated in Fig. \ref{mrnaall}(b), which has a ladder-shaped structure. In fact, ladder-shaped models arise in many gene expression models and have been extensively studied in the literature \cite{peccoud1995markovian, shahrezaei2008analytical, zhou2012analytical, hornos2005self, grima2012steady, vandecan2013self, kumar2014exact, bokes2015protein, jia2020small}. Such models are usually analytically tractable with their solutions being represented by hypergeometric functions (see \cite{melykuti2014equilibrium} for a detailed discussion on the analytical theory of ladder-shaped models). Note that the irreducible state space of the original model is the whole three-dimensional nonnegative integer lattice since $m_\star$, $m$, and $n$ can take all nonnegative integer values. Using the method proposed in this paper, we simplify a three-dimensional problem for the original model to a coupled one-dimensional problem for the modified model (here ``coupled" means that $m$ can only take the values of $0$ and $1$ and ``one-dimensional" means that $n$ ranges over all nonnegative integers), which greatly reduces the theoretical complexity.

Since the modified model is essentially one-dimensional, its generating function $H$ can be easily computed in steady-state conditions, which is given by (see Appendix C for details)
\begin{equation*}\label{gfinal2}
{\small\begin{split}
H(x_\star,x,y) = a\pi_{\mathbf{0}}(x_\star-1)+b\pi_{\mathbf{0}}
\left[(x-1){}_1F_1(1;1+\nu;\mu(y-1))+\mu\int_1^y{}_{1}F_1(1;1+\nu;\mu(z-1))dz\right]+1,
\end{split}}
\end{equation*}
where ${}_1F_1$ denotes the confluent hypergeometric function and
\begin{equation*}
a = \frac{s}{k+f},\;\;\;b = \frac{ks}{(k+f)v},\;\;\;\mu = \frac{u}{d},\;\;\;\nu = \frac{v}{d}.
\end{equation*}
It then follows from Eq. \eqref{expression} that the generating function $F$ of the original model is given by
\begin{equation}\label{preexpression}
\begin{aligned}
F(x_\star,x,y) = e^{a(x_\star-1)+b\left[(x-1){}_1F_1(1;1+\nu;\mu(y-1))
+\mu\int_1^y{}_{1}F_1(1;1+\nu;\mu(z-1))dz\right]}.
\end{aligned}
\end{equation}
This implies that the number of nascent mRNA is independent of the numbers of mature mRNA and protein, while the numbers of mature mRNA and protein are correlated. Taking the derivatives of $F$ gives the joint distributions for the nascent mRNA, mature mRNA, and protein numbers. In particular, taking $x=y=1$ and $x_\star=y=1$, we obtain
\begin{equation*}
F(x_\star,1,1) = e^{a(x_\star-1)},\;\;\;F(1,x,1) = e^{b(x-1)}.
\end{equation*}
This shows that the numbers of nascent and mature mRNAs both have a Poisson distribution:
\begin{equation*}
p^{M^{\star}}_{m_\star} = \frac{a^{m_\star}}{m_\star!}e^{-a},\;\;\;p^M_m = \frac{b^m}{m!}e^{-b}.
\end{equation*}
Taking $x_\star=x=1$, we obtain
\begin{equation}\label{premrnaexpresion}
F(1,1,y) = e^{b\mu\int_1^y{}_{1}F_1(1;1+\nu;\mu(z-1))dz}.
\end{equation}
It then follows from Eq. \eqref{marginal} that the marginal distribution for the protein number is given by
\begin{equation}\label{expression1}
p_n^P = \frac{B_{n}(g_1,\dots,g_{n})}{n!}e^{-b\mu\int_0^1{}_{1}F_1(1;1+\nu;\mu(z-1))dz},
\end{equation}
where $B_n$ is the complete Bell polynomial and
\begin{equation*}
g_i = \frac{b\mu^{i}(i-1)!}{(1+\nu)_{i-1}}{}_{1}F_1(i;i+\nu;-\mu),\;\;\; i=1,\dots,n.
\end{equation*}

Our analytical results can also be used to compute the correlation coefficient $\rho_{M,P}$ between the mature mRNA and protein numbers. Combining Eqs. \eqref{correlation} and \eqref{preexpression}, it is easy to obtain
\begin{equation*}
\rho_{M,P} = \sqrt{\frac{\mu}{(1+\nu)(1+\mu+\nu)}}.
\end{equation*}
This shows that the mature mRNA and protein numbers are always positively correlated; the correlation is strong when the translation rate $u$ is large and the mature mRNA degradation rate $v$ is small compared to the protein decay rate $d$.

We next take a deeper look at the marginal protein distribution. It is a classical result that if mRNA decays much faster than protein, then the protein number has a negative binomial distribution \cite{shahrezaei2008analytical}. In fact, this assumption holds for the majority of genes in bacteria and yeast, but it fails for many genes in higher prokaryotes, where mRNA and protein often decay at the same time scale (see Table S1 in \cite{jia2021frequency} for such time scales in various cell types). Here we consider another important case where mature mRNA decays much slower compared to protein. Specifically, we consider the limiting case of $\nu = v/d \ll 1$, while keeping $b = ks/(k+f)v$ as constant. In this limit, the synthesis and degradation of mature mRNA are both very slow and thus the mature mRNA is a slow variable. Actually, a similar limit has been considered in \cite{bokes2012exact} where the nascent mRNA is not modeled explicitly; here we take a deeper look at  this limit. Since $\nu\ll 1$, we have ${}_{1}F_1(1;1+\nu;\mu(z-1)) \approx {}_{1}F_1(1;1;\mu(z-1)) = e^{\mu(z-1)}$. Then the generating function in Eq. \eqref{premrnaexpresion} can be simplified as
\begin{equation*}
F(1,1,y) = e^{b\left[e^{\mu(y-1)}-1\right]}.
\end{equation*}
Taking the derivatives of the generating function $F(1,1,y)$ at $y=0$, we find that the protein number has the following mixed Poisson distribution with Poissonian weights:
\begin{equation}\label{mixture2}
p^P_n = e^{-b}\delta_{0}(n)
+\sum_{k=1}^\infty\frac{b^je^{-b}}{k!}\left[\frac{\left(k\mu\right)^ne^{-k\mu}}{n!}\right],
\end{equation}
where $\delta_{0}(n)$ is Kronecker's delta function which takes the value of $1$ when $n=0$ and the value of $0$ otherwise. This can be understood intuitively as follows. We have seen that the mature mRNA number has the Poisson distribution $\Pnum(N_M = k) = b^ke^{-b}/k!$, where $N_M$ denotes the number of $M$. Since the mature mRNA is a slow variable, given that $k$ copies of mature mRNA has been produced, the total synthesis rate of protein is given by $ku$ and thus the conditional distribution of the protein number is also Poissonian:
\begin{equation*}
\Pnum(N_P = n|N_M = k) = \frac{\left(k\mu\right)^n e^{-k\mu}}{n!},
\end{equation*}
where $N_P$ denotes the number of $P$. Hence the mixed Poisson distribution given in Eq. \eqref{mixture2} is nothing but the formula of total probability:
\begin{equation*}
\Pnum(N_P = n) = \sum_{k=0}^\infty\Pnum(N_P = n|N_M = k)\Pnum(N_M = k).
\end{equation*}
Fig. \ref{mrnaall}(c) shows the comparison between the exact solution given in Eq. \eqref{expression1}, the approximate solution given in Eq. \eqref{mixture2}, and FSP simulations under different values of $\nu$. It can be seen that the exact and approximation solutions coincide perfectly with each other for small $\nu$, but they fail as expected for large $\nu$. When $\nu\ll 1$, the protein distribution is a mixture of Poisson distributions and thus is capable of producing multiple peaks that are located around $k\mu$, $k = 0,1,2...$ with $\mu$ being the averaged amount of protein produced by a single mature mRNA molecule. Note that only the first several Poisson components contribute to the multiple peaks of the protein distribution since the Poisson components become lower and flatter as $k$ increases (Fig. \ref{mrnaall}(c)). In the literature \cite{jia2020small}, it is widely believed that bimodality of the protein distribution has two major origins --- it can occur either when there is a positive feedback loop involved in the system or when the switching between promoter states are slow. Here we show that multimodality can also be caused by slow synthesis and degradation of mature mRNA, even when the gene is constitutively expressed (no promoter switching). As $\nu$ increases, multimodality disappears and the protein distribution becomes closer to a negative binomial distribution (Fig. \ref{mrnaall}(c)).

\subsection{Gene regulatory model with translational bursting}\label{cascadesection}
As the third application, we consider a simple gene regulatory system where the product of a gene, as a transcription factor, regulates the expression of another gene in a bursty manner (Fig. \ref{cascadeall}(a)). Let $G_1$ and $G_2$ denote the two genes and let $P_1$ and $P_2$ denote the corresponding gene products. The effective reactions describing the gene regulatory system are given by
\begin{gather*}
G_1 \xrightarrow{u_1} G_1+P_1,\;\;\; G_2+P_1 \xrightarrow{u_2p^kq}G_2+P_1+kP_2,\;\;\; k\geq 1,\\
P_1 \xrightarrow{d_1} \varnothing,\;\;\; P_2\xrightarrow{d_2}\varnothing.
\end{gather*}
Here the first reaction describes the expression of gene $G_1$ with effective translation rate $u_1$, the second reaction describes the expression of gene $G_2$ which is activated by protein $P_1$, and the last two reactions describe the degradation of the two proteins. In agreement with experiments \cite{cai2006stochastic}, the production of protein $P_2$ is assumed to occur in bursts of random size sampled from a geometric distribution with parameter $p$. Each burst is due to rapid synthesis of protein from a single, short-lived mRNA molecule; thus the effective translation rate of gene $G_2$ is the product of the corresponding transcription rate $u_2$ and the geometric distribution $p^kq$, where $q = 1-p$ \cite{jia2017simplification}. The microstate of the system can be represented by an ordered pair $(n_1,n_2)$, where $n_i$ denotes the copy number of protein $P_i$. Let $p_{n_1,n_2}$ denote the probability of observing microstate $(n_1,n_2)$ and let
\begin{equation*}
F(y_1,y_2)=\sum_{n_1,n_2}p_{n_1,n_2}y_1^{n_1} y_2^{n_2}
\end{equation*}
denote the corresponding generating function. Then the evolution of the gene regulatory system is governed by the CME
\begin{equation*}
\begin{split}
\dot{p}_{n_1,n_2} &= u_1p_{n_1-1,n_2}+\sum_{i=0}^{n_2-1}u_2p^{n_2-i}qn_1p_{n_1,i}
+d_1(n_1+1)p_{n_1+1,n_2}+d_2(n_2+1)p_{n_1,n_2+1}\\
&\quad-(u_1+u_2pn_1+d_1n_1+d_2n_2)p_{n_1,n_2},
\end{split}
\end{equation*}
where $u_2pn_1 = \sum_{k=1}^\infty u_2p^kqn_1$ in the bracket is the sum of transition rates from microstate $(n_1,n_2)$ to other microstates due to translational bursting.

To solve this CME, we next consider the modified Markovian model. Similarly, we do not take the copy number variation of the gene into account and thus the reaction $G_1\rightarrow G_1+P_1$ can be viewed as a zero-order reaction. Since the zero-order reaction can only occur at the zero microstate, it is easy to see that the irreducible state space of the modified model is given by
\begin{equation*}
\{(1,n_2),(0,n_2):\;n_2\geq 0\},
\end{equation*}
and the transition diagram restricted to the irreducible state space is illustrated in Fig. \ref{cascadeall}(b). Note that the irreducible state space of the original model is the two-dimensional lattice since $n_1$ and $n_2$ can take all nonnegative integer values. Thus the method proposed in this paper reduces a two-dimensional problem for the original model to a coupled one-dimensional problem for the modified model (here ``coupled" means that $n_1$ can only take the values of $0$ and $1$), which greatly reduces the theoretical complexity.
\begin{figure}[!htb]
\centering\includegraphics[width=1.0\textwidth]{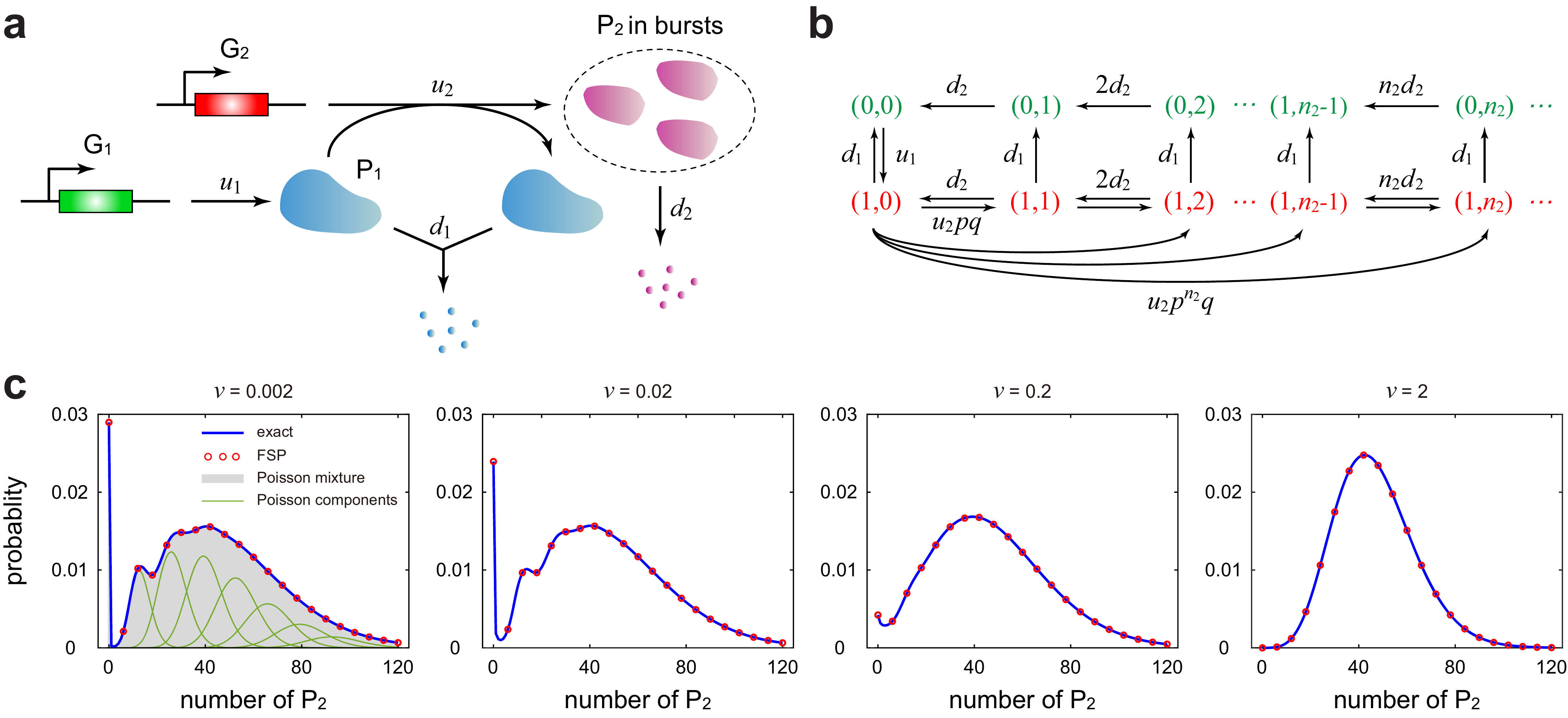}
\caption{\textbf{A gene regulatory model with translational bursting.} (a) Schematic of a simple gene regulatory model where the product of gene $G_1$ activates the expression of gene $G_2$. The protein synthesis of gene $G_2$ occurs in bursts. (b) Transition diagram for the modified model restricted to the irreducible state space. Note that translational bursting can cause jumps from microstate $(1,n_2)$ to $(1,n_2')$ with $n_2'>n_2$. This is shown for microstate $(1,0)$ in the figure but is also true for other microstates. (c) Comparison between the exact steady-state distribution for the number of protein $P_2$ given in Eq. \eqref{expression2} (blue curve) with FSP simulations (red circles) as $\nu = d_1/d_2$ varies while keeping $\mu_1 = u_1/d_1$ as constant. The left panel also compares the exact solution (blue curve) with mixed negative binomial approximation given in Eq. \eqref{mixture4} (grey region). The negative binomial components of the mixed distribution are shown by the green curves. The model parameters are chosen as $u_2=40, d_2=1, p=0.25$ and the parameters $u_1$ and $d_1$ are chosen so that $\mu_1=3.52$.}\label{cascadeall}
\end{figure}

Since the modified model is essentially one-dimensional, its generating function $H$ can be easily computed in steady-state conditions, which is given by (see Appendix D for details)
\begin{equation*}\label{gfinal3}
{\small\begin{split}
	H(y_1,y_2) =&\; \pi_{\mathbf{0}}\mu_1\bigg[{}_2F_1\left(-\mu_2,1;1+\nu;\omega(y_2)\right)(y_1-1)\\
	&\;+\mu_2 B\int_1^{y_2}{}_2F_1\left(1+\mu_2+\nu,1;1+\nu;B(z-1)\right)dz\bigg]+1,
	\end{split}}
\end{equation*}
where ${}_{2}F_1$ denotes the Gaussian hypergeometric function, $B = p/q = \sum_{n_2=1}^\infty n_2p^{n_2}q$ is the mean burst size of protein $P_2$, and
\begin{equation*}
\mu_1 = \frac{u_1}{d_1},\;\;\;\mu_2 = \frac{u_2}{d_2},\;\;\;\nu = \frac{d_1}{d_2},\;\;\;
\omega(y_2) = \frac{p(y_2-1)}{py_2-1}.
\end{equation*}
It then follows from Eq. \eqref{expression} that the generating function $F$ for the original model is given by
\begin{equation}\label{cascadeexpression}
\begin{aligned}
F(y_1,y_2) = e^{\mu_1\left[{}_2F_1\left(-\mu_2,1;1+\nu;\omega(y_2)\right)(y_1-1)+\mu_2 B\int_1^{y_2}{}_2F_1\left(1+\mu_2+\nu,1;1+\nu;B(z-1)\right)dz\right]}.
\end{aligned}
\end{equation}

Taking the derivatives of $F$ at zero yields the steady-state joint distribution for the numbers of the two proteins. In particular, taking $y_2=1$, we obtain $F(y_1,1) = e^{\mu_1(y_1-1)}$. This shows that the number of protein $P_1$ has the Poisson distribution
\begin{equation*}
p^{P_1}_{n_1} = \frac{\mu_1^{n_1}}{n_1!}e^{-\mu_1}.
\end{equation*}
Moreover, taking $y_1=1$, we obtain
\begin{equation}\label{marginal3}
F(1,y_2) = e^{\mu_1\mu_2 B\int_1^{y_2}{}_2F_1\left(1+\nu+\mu_2,1;1+\nu;B(z-1)\right)dz}.
\end{equation}
It then follows from Eq. \eqref{marginal} that the number of protein $P_2$ has the following distribution:
\begin{equation}\label{expression2}
p_{n_2}^{P_2} = \frac{B_{n_2}(g_1,\dots,g_{n_2})}{n_2!}
e^{-\mu_1\mu_2 B\int_0^{1}{}_2F_1\left(1+\nu+\mu_2,1;1+\nu;B(z-1)\right)dz},
\end{equation}
where $B_n$ is the complete Bell polynomial and
\begin{equation*}
g_i = \frac{\mu_1\mu_2B^i(1+\nu+\mu_2)_{i-1}(i-1)!}{(1+\nu)_{i-1}}
{}_{2}F_1\left(i+\nu+\mu_2,i;i+\nu;-B\right),\;\;\; i=1,\dots,n_2.
\end{equation*}

We next focus on two limiting cases. The first case occurs when protein $P_1$ decays much faster than protein $P_2$, i.e. $\nu = d_1/d_2\gg 1$, and the constant $\mu_1 = u_1/d_1$ is strictly positive and bounded. In this case, both the synthesis and degradation of protein $P_1$ are very fast and thus it can be viewed as a fast variable. When $\nu\gg 1$, we have
\begin{equation*}
{}_2F_1\left(1+\nu+\mu_2,1;1+\nu;B(z-1)\right) \approx {}_1F_0\left(1;B(z-1)\right) = (1-B(z-1))^{-1}.
\end{equation*}
It then follows from Eq. \eqref{marginal3} that
\begin{equation}
F(1,y_2) = \left(\frac{q}{1-py_2}\right)^{\mu_1\mu_2}.
\end{equation}
This shows that the number of protein $P_2$ has the negative binomial distribution
\begin{equation}\label{nb1}
p^{P_2}_{n_2} = \frac{\left(\mu_1\mu_2\right)_{n_2}}{n_2!}p^{n_2}q^{\mu_1\mu_2}.
\end{equation}

The second case occurs when protein $P_1$ decays much slower than protein $P_2$, i.e. $\nu = d_1/d_2\ll 1$, and the constant $\mu_1 = u_1/d_1$ is strictly positive and bounded. In this case, both the synthesis and degradation of protein $P_1$ are very slow and thus it can be viewed as a slow variable. When $\nu\ll 1$, we have
\begin{equation*}
{}_2F_1\left(1+\nu+\mu_2,1;1+\nu;B(z-1)\right)
\approx {}_1F_0\left(1+\mu_2;B(z-1)\right) = \left(1-B(z-1)\right)^{-(1+\mu_2)}.
\end{equation*}
Then the generating function in Eq. \eqref{marginal3} can be simplified as
\begin{equation*}
F(1,y_2) = e^{\mu_1\left[\left(\frac{q}{1-py_2}\right)^{\mu_2}-1\right]}.
\end{equation*}
Taking the derivatives of $F(1,y_2)$ at $y_2=0$, we find that the number of protein $P_2$ has the following mixed negative binomial distribution with Poissonian weights:
\begin{equation}\label{mixture4}
p^{P_2}_{n_2} = e^{-\mu_1}\delta_{0}(n_2)
+\sum_{k=1}^\infty\frac{\mu_1^ke^{-\mu_1}}{k!}\left[\frac{(k\mu_2)_{n_2}}{n_2!}p^{n_2}q^{k\mu_2}\right].
\end{equation}
This can be explained intuitively as follows. We have seen that the number of protein $P_1$ has the Poisson distribution $\Pnum(N_{P_1}=k) = \mu_1^ke^{-\mu_1}/k!$, where $N_{P_1}$ denotes the number of $P_1$. Since protein $P_1$ is a slow variable, given that $k$ copies of $P_1$ has been produced, the effective transcription rate of gene $G_2$ is given by $ku_2$ and thus the conditional distribution for the number of protein $P_2$ is negative binomial:
\begin{equation*}
\Pnum(N_{P_2}=n_2|N_{P_1}=k) = \frac{(k\mu_2)_{n_2}}{n_2!}p^{n_2}q^{k\mu_2},
\end{equation*}
where $N_{P_2}$ denotes the number of $P_2$. Thus the mixed negative binomial distribution given in Eq. \eqref{mixture4} is nothing but the formula of total probability:
\begin{equation*}
\Pnum(N_{P_2}=n_2) = \sum_{k=0}^\infty\Pnum(N_{P_2}=n_2|N_{P_1}=k)\Pnum(N_{P_1}=k).
\end{equation*}
Fig. \ref{cascadeall}(c) shows the comparison between our exact solution given in Eq. \eqref{expression2}, the approximate solution given in Eq. \eqref{mixture4}, and FSP simulations under different values of $\nu$. Clearly, the exact and approximation solutions coincide perfectly with each other for small $\nu$, but deviate significantly from each other for large $\nu$. When $\nu\ll 1$, the copy number distribution for protein $P_2$ is a mixture of negative binomials and thus can produce multiple peaks around $k\mu_2B$, $k = 0,1,2...$ with $\mu_2B$ being the averaged amount of protein $P_2$ produced by a single protein $P_1$ molecule. As $\nu$ increases, multimodality disappears and the protein distribution becomes closer to a negative binomial distribution (Fig. \ref{cascadeall}(c)).

We finally examine the correlation between the two proteins using our analytical results. It follows from Eqs. \eqref{correlation} and \eqref{cascadeexpression} that the correlation coefficient between the numbers of $P_1$ and $P_2$ is given by
\begin{equation*}
\rho_{P_1,P_2}=\sqrt{\frac{\mu_2B}{(1+\nu)[1+\nu+(1+\nu+\mu_2)B]}}.
\end{equation*}
Clearly, the numbers of the two proteins are always positively correlated; the correlation is strong when the burstiness of protein $P_2$ is large, the translation rate of protein $P_2$ is large, and the degradation rate of protein $P_1$ is small.

\subsection{Gene expression model with alternative splicing}\label{alsection}
Alternative splicing is a process during gene expression that results in a single gene coding for multiple proteins \cite{saitou2013introduction}. In this process, particular exons of a gene may be included within or excluded from the final processed mRNA that are produced from that gene. Consequently, the proteins translated from different spliced mRNAs will be different (see Fig. \ref{al1}(a) for an illustration). A gene expression model involving alternative splicing has been solved in \citep{wang2014alternative}, which considers the expression of mRNAs but not proteins. Here we take proteins into consideration.

Let $G$ denote the gene of interest, $M_\star$ denote the nascent mRNA, $M_1$ and $M_2$ denote two mRNA isoforms, and $P_1$ and $P_2$ denote the corresponding protein isoforms. Based on the central dogma of molecular biology, the effective reactions involved in the gene expression system are listed as follows:
\begin{equation}\label{alternativescheme}
\begin{aligned}
&G \xrightarrow{s} G+M_\star, \;\;\;  M_\star\xrightarrow{k_i} M_i, \;\;\;  M_i\xrightarrow{u_i}M_i+P_i,\\
&\quad M_\star\xrightarrow{f}\varnothing, \;\;\; M_i\xrightarrow{v_i}\varnothing,\;\;\; P_i
\xrightarrow{d_i}\varnothing, \;\;\; i = 1,2,
\end{aligned}
\end{equation}
where $s$ is the transcription rate, $k_i$ are the rates of alternative splicing, $u_i$ are the translation rates of the two mRNA isoforms, and $f$, $v_i$, and $d_i$ are the degradation rates of all gene products. Recently, it has been found that alternative splicing can be regulated by a system of proteins (regulators) binding to a nascent transcript that in turn direct the splicing machinery to include or skip specific exons \cite{ajith2016position-dependent, Fu2014Context}; moreover, the regulators usually exert distinct effects on exon inclusion or exclusion depending on the position of its binding \citep{ajith2016position-dependent} and thus different binding positions lead to different mRNA isoforms. Here we take this effect into account by assuming that there is a regulator $P$ which activates the formation of two mRNA isoforms $M_1$ and $M_2$ (via exon inclusion and/or exclusion). Hence the copy number of the regulator $P$, which is denoted by $n$, will influence the splicing rates $k_1=k_1(n)$ and $k_2=k_2(n)$. For simplicity, we further assume that the number of regulator has a fixed distribution that is independent of the numbers of gene products.
\begin{figure}[!htb]
\centering\includegraphics[width=1.0\textwidth]{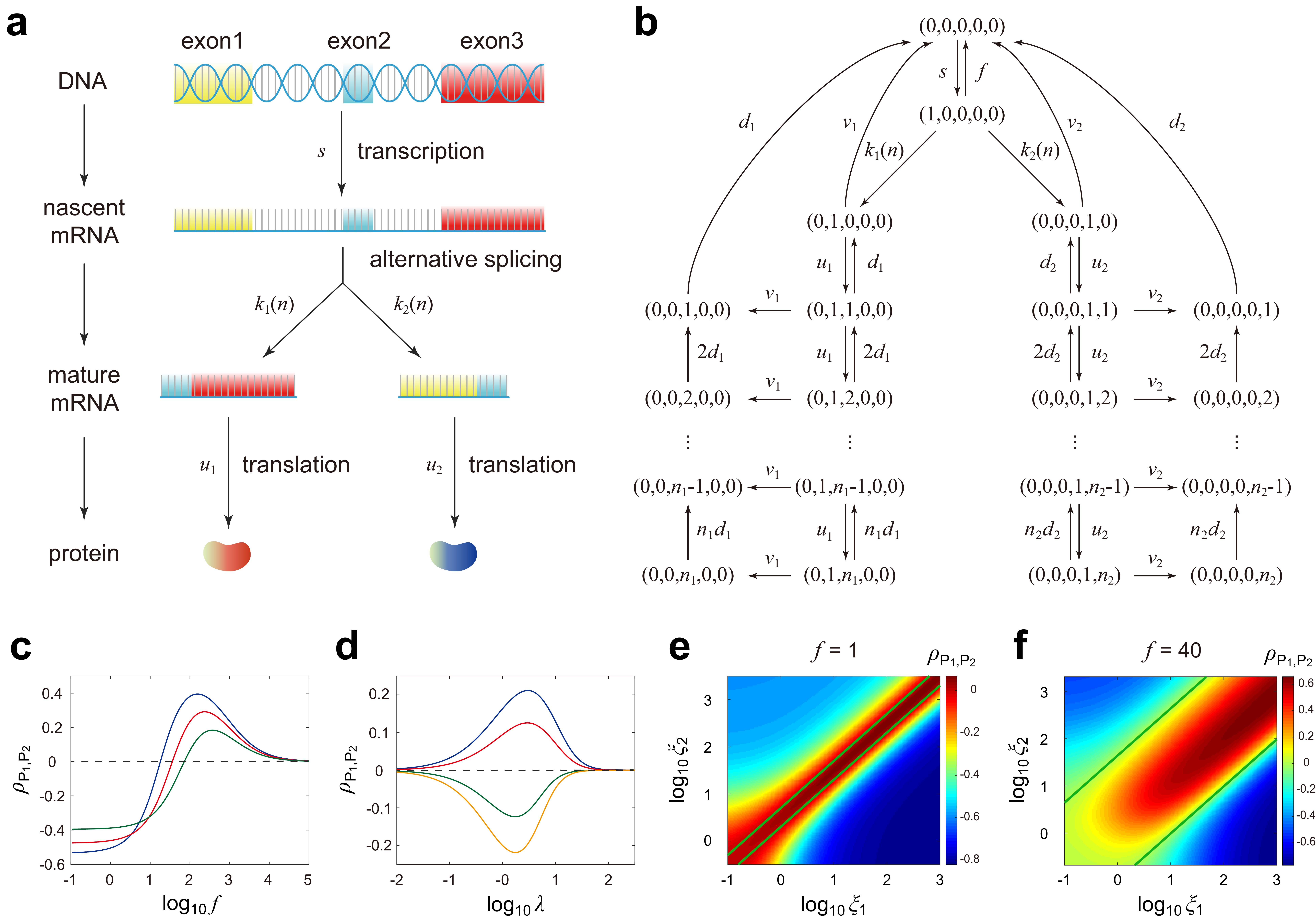}
\caption{\textbf{A multi-step gene expression model with alternative splicing.} (a) Schematic of a multi-step gene expression model involving transcription, translation, and alternative splicing. Due to alternative splicing, the nascent mRNA is spliced in two different ways to produce two mature mRNA isoforms. (b) Transition diagram for the modified model restricted to the irreducible state space. Note that the transition diagram has two branches (left and right), corresponding to the production of two mRNA/protein isoforms. Each branch has a ladder-shaped structure. (c) Correlation coefficient $\rho_{P_1,P_2}$ between the numbers of the two protein isoforms versus the degradation rate $f$ of nascent mRNA. The model parameters are chosen as $s=100, u_1=30, u_2=20, v_1=15, v_2=4, d_1=3, d_2=2, \xi_1=10, \xi_2=20, \lambda=2$. The remaining parameters are chosen as $\eta_1 = 10, \eta_2 = 3$ (blue curve), $\eta_1 = 20, \eta_2 = 6$ (red curve), and $\eta_1 = 40, \eta_2 = 12$ (green curve). (d) Correlation coefficient $\rho_{P_1,P_2}$ versus the mean $\lambda$ of the regulator number. The model parameters are chosen as $u_1=30, u_2=20, v_1=3, v_2=4, d_1=3, d_2=4, \xi_1=7, \xi_2=5, \eta_1=15, \eta_2=28$. The remaining parameters are chosen as $s = 200, f = 80$ (blue curve), $s = 100, f = 80$ (red curve), $s = 100, f = 0$ (green curve), and $s = 200, f = 0$ (orange curve). (e),(f) Correlation coefficient $\rho_{P_1,P_2}$ versus the regulation strengths $\xi_1$ and $\xi_2$. (e) Slow degradation of nascent mRNA with $f = 1$. (f) Fast degradation of nascent mRNA with $f = 40$. The model parameters are chosen as $s=100, u_1=30, u_2=20, v_1=15, v_2=4, d_1=3, d_2=2, \eta_1=1, \eta_2=4, \lambda=2$. The two green lines separate the region with positive correlation and the region with negative correlation.}\label{al1}
\end{figure}

The microstate of the system can be represented by an ordered five-tuple $(m,m_1,n_1,m_2,n_2)$: the copy number $m$ of nascent mRNA, the copy numbers $m_1$ and $m_2$ of the two mRNA isoforms, and the copy numbers $n_1$ and $n_2$ of the two protein isoforms. Let $p_{m,m_1,n_1,m_2,n_2}$ denote the probability of observing microstate $(m,m_1,n_1,m_2,n_2)$ and let
\begin{equation*}
F(x,x_1,y_1,x_2,y_2) = \sum_{m,m_1,n_1,m_2,n_2}p_{m,m_1,n_1,m_2,n_2}x^{m}x_1^{m_1}y_1^{n_1}x_2^{m_2}y_2^{n_2}
\end{equation*}
denote the corresponding generating function. Given that there are $n$ copies of regulator $P$, we can treat the splicing rates $k_1 = k_1(n)$ and $k_2 = k_2(n)$ as constants and the evolution of the gene expression model is governed by the CME
\begin{equation*}
\begin{aligned}
&\;\dot{p}_{m,m_1,n_1,m_2,n_2}\\
=&\;sp_{m-1,m_1,n_1,m_2,n_2}+k_1(m+1)p_{m+1,m_1-1,n_1,m_2,n_2}+k_2(m+1)p_{m+1,m_1,n_1,m_2-1,n_2}\\
&\;+u_1m_1p_{m,m_1,n_1-1,m_2,n_2}+u_2m_2p_{m,m_1,n_1,m_2,n_2-1}+f(m+1)p_{m+1,m_1,n_1,m_2,n_2}\\
&\;+v_1(m_1+1)p_{m,m_1+1,n_1,m_2,n_2}+v_2(m_2+1)p_{m,m_1,n_1,m_2+1,n_2}\\
&\;+d_1(n_1+1)p_{m,m_1,n_1+1,m_2,n_2}+d_2(n_2+1)p_{m,m_1,n_1,m_2,n_2+1}\\
&\;-[s+(k_1+k_2+f)m+(u_1+v_1)m_1+(u_2+v_2)m_2+d_1n_1+d_2n_2]p_{m,m_1,n_1,m_2,n_2}.
\end{aligned}
\end{equation*}

To solve this CME, we next consider the modified Markovian model, which has only one zero-order reaction. Since the zero-order reaction $G\rightarrow G+M_\star$ can only occur at the zero microstate, it is easy to see that the irreducible state space of the modified model is given by
\begin{equation*}
\{(1,0,0,0,0),(0,1,n_1,0,0),(0,0,n_1,0,0),(0,0,0,1,n_2),(0,0,0,0,n_2):\;n_1,n_2\geq 0\}.
\end{equation*}
The transition diagram restricted to the irreducible state space is illustrated in Fig. \ref{al1}(b). Clearly, the zero microstate can only transition to microstate $(1,0,0,0,0)$. If the nascent transcript $M_\star$ produces the mRNA isoform $M_1$, then the modified model enters the left branch in Fig. \ref{al1}(b); if $M_\star$ produces $M_2$, then the modified model enters the right branch. Hence our method reduces a five-dimensional problem for the original model to a coupled one-dimensional problem for the modified model. In analogy to the derivation in Section \ref{presection}, given that there are $n$ copies of regulator $P$ in a single cell, the generating function of the original model is given by (see Appendix E for details)
\begin{equation}\label{final4}
\begin{aligned}
&\; F(x,x_1,y_1,x_2,y_2|n)\\
=&\; e^{a(n)(x-1)+\sum_{i=1}^2K_i(n)b_i\left[(x_i-1){}_1F_1(1;1+\nu_i;\mu_i(y_i-1))
	+\mu_i\int_1^{y_i}{}_1F_1(1;1+\nu_i;\mu_i(z-1))dz\right]}.
\end{aligned}
\end{equation}
where
\begin{gather*}
a(n)=\frac{s}{k_1(n)+k_2(n)+f},\;\;\;
K_1(n)=\frac{k_1(n)}{k_1(n)+k_2(n)+f},\;\;\;K_2(n)=\frac{k_2(n)}{k_1(n)+k_2(n)+f},\\
b_1=\frac{s}{v_1},\;\;\;b_2=\frac{s}{v_2},\;\;\;
\mu_1=\frac{u_1}{d_1},\;\;\;\mu_2=\frac{u_2}{d_2},\;\;\;
\nu_1=\frac{v_1}{d_1},\;\;\;\nu_2=\frac{v_2}{d_2}.
\end{gather*}
Finally, when taking into account the copy number variation of regulator $P$, it follows from the total probability formula that the generating function $F$ is given by
\begin{equation*}
F(x,x_1,y_1,x_2,y_2)
=\sum_{n=0}^\infty p^P_nF(x,x_1,y_1,x_2,y_2|n),
\end{equation*}
where $p^P_n$ is the probability of observing $n$ copies of regulator in a cell. Finally, the joint distribution of all gene products can be recovered by taking the derivatives of the generating function. It is easy to see that the marginal distributions for nascent mRNA and the two mRNA isoforms are all mixed Poisson distributions with the weights being the distribution of the regulator number; however, the marginal distributions for the two protein isoforms are much more complicated.

In recent years, the correlation between different mRNA and protein species produced from a single gene by means of alternative splicing has attracted increasing attention \cite{2017Alternative, ajith2016position-dependent}. It has been shown that the numbers of two mRNA isoforms are independent of each other if they are not controlled by the regulator \cite{wang2014alternative}; moreover, transcriptional bursting (which is not considered in our current model) may lead to positive correlation between two mRNA isoforms \cite{wang2014alternative}. Here we analyze such correlation when the two mRNA isoforms are controlled by the same regulator. To do this, we assume that the splicing rates depend on the regulator number linearly as
\begin{equation*}
k_1(n) = \xi_1n+\eta_1,\;\;\;k_2(n) = \xi_2n+\eta_2,
\end{equation*}
where $\eta_i>0$, $i = 1,2$ are the spontaneous splicing rates and $\xi_i\geq 0$ characterize the strengths of regulation. Such linear dependence has been widely used in the modeling of stochastic gene regulatory networks \cite{kumar2014exact, jia2017emergent, jia2017stochastic, jia2019single}. In addition, we assume that the number of regulator has a Poisson distribution with mean $\lambda$. Under these assumptions, the correlation coefficient between (the numbers of) the two mRNA isoforms is given by (see Appendix E for details)
\begin{equation}\label{cor1}
\rho_{M_1,M_2} = \frac{\alpha_1\alpha_2}{\sqrt{(\alpha_1^2+\beta_1)(\alpha_2^2+\beta_2)}},
\end{equation}
and the correlation coefficient between the two protein isoforms is given by (see Appendix E for details)
\begin{equation}\label{cor2}
\rho_{P_1,P_2} = \frac{\alpha_1\alpha_2}{\sqrt{(\alpha_1^2+\beta_1L_1)(\alpha_2^2+\beta_2L_2})},
\end{equation}
where
\begin{gather*}
\alpha_1=\frac{\xi_2\eta_1-\xi_1\eta_2-\xi_1f}{\eta_1+\eta_2+f},\;\;\; \alpha_2=\frac{\xi_1\eta_2-\xi_2\eta_1-\xi_2f}{\eta_1+\eta_2+f},\\
L_1=\frac{1+\mu_1+\nu_1}{\mu_1(1+\nu_1)},\;\;\;
L_2=\frac{1+\mu_2+\nu_2}{\mu_2(1+\nu_2)},\;\;\;
\gamma=\frac{\eta_1+\eta_2+f}{\xi_1+\xi_2},\\
h_1={}_1F_1\left(1;\gamma+1;-\lambda\right),\;\;\; h_2={}_2F_2\left(\gamma,\gamma;\gamma+1,\gamma+1;\lambda\right)e^{-\lambda},\\
\beta_1=\frac{(\xi_1+\xi_2)(\xi_1+\alpha_1h_1)}{b_1(h_2-h_1^2)},\;\;\; \beta_2=\frac{(\xi_1+\xi_2)(\xi_2+\alpha_2h_1)}{b_2(h_2-h_1^2)},
\end{gather*}
where ${}_2F_2$ denotes the generalized hypergeometric function. In the above formulas, the parameters $\beta_1$ and $\beta_2$ depend on the parameters $h_1$ and $h_2$, which further depend on the parameter $\gamma$. In Appendix E, we have proved that the parameters $\beta_1$ and $\beta_2$, together with $h_2-h_1^2$, must be positive. Therefore, the correlation coefficients $\rho_{M_1,M_2}$ and $\rho_{P_1,P_2}$ must have the same sign and the sign is determined by the sign of $\alpha_1\alpha_2$. In particular, when the nascent mRNA decays very slowly, i.e. $f\ll 1$, we have
\begin{equation*}
\alpha_1\alpha_2 \approx -\frac{(\xi_2\eta_1-\xi_1\eta_2)^2}{(\eta_1+\eta_2)^2}.
\end{equation*}
In this case, the numbers of the two mRNA/protein isoforms are negatively correlated. On the other hand, when the nascent mRNA decays very fast, i.e. $f\gg 1$, we have $\alpha_1\alpha_2 \approx \xi_1\xi_2 > 0$. In this case, the numbers of the two mRNA/protein isoforms are positively correlated.

These results can be understood intuitively as follows. When the nascent mRNA decays very slowly, once a nascent transcript is synthesized, it can either produce an $M_1$ or an $M_2$ molecule. Thus there is strong competition between the two isoforms; the more one isoform, the less the other isoform. This results in negative correlation between them. On the other hand, when the nascent mRNA decays very fast, its molecule number relaxes to the steady-state value rapidly \cite{jia2018relaxation} and thus there is an ample supply of nascent mRNA. In this case, there is little competition between the two isoforms; the more (less) the regulator, the more (less) the two isoforms. This results in positive correlation between them.

Our results indicate that the degradation rate $f$ of nascent mRNA has a critical value
\begin{equation*}
f_c = \begin{cases}
|\xi_2\eta_1-\xi_1\eta_2|/\xi_1, &\textrm{if}\;\xi_2\eta_1-\xi_1\eta_2\geq 0,\\
|\xi_2\eta_1-\xi_1\eta_2|/\xi_2, &\textrm{if}\;\xi_2\eta_1-\xi_1\eta_2<0,
\end{cases}
\end{equation*}
and the system undergoes a stochastic bifurcation as $f$ varies. When $f<f_c$, we have $\alpha_1\alpha_2<0$ and thus the levels of the two isoforms are negatively correlated; when $f=f_c$, we have $\alpha_1\alpha_2 = 0$ and thus they are not correlated; when $f>f_c$, we have $\alpha_1\alpha_2 > 0$ and thus they are positively correlated. Note that the size of the critical value $f_c$ depends on the sizes of $\eta_1$ and $\eta_2$. As $\eta_1$ and $\eta_2$ increase, the critical value $f_c$ becomes larger. These observations coincide with stochastic simulations in Fig. \ref{al1}(c), which illustrates the correlation coefficient $\rho_{P_1,P_2}$ as a function of $f$.

The correlation between the two mRNA/protein isoforms is also influenced by the abundance of regulator. Fig. \ref{al1}(d) depicts the correlation coefficient $\rho_{P_1,P_2}$ as a function of the regulator mean $\lambda$. It can be seen that the correlation is weak when $\lambda$ is very small or very large. Interestingly, there is an optimal $\lambda$ such that $|\rho_{P_1,P_2}|$ attains its maximum. This shows that the correlation is the strongest when the regulator mean is neither too small nor too large. This can be understood intuitively as follows. It follows from Eq. \eqref{final4} that the regulator number $n$ affects the joint distribution by adjusting the three parameters $a(n)$, $K_1(n)$, and $K_2(n)$. When $\lambda\ll 1$ or $\lambda\gg 1$, the three parameters are almost invariant and thus the gene expression model under consideration behaves like a system with no regulator. This explains the weak correlation observed when $\lambda\ll 1$ or $\lambda\gg 1$. Fig. \ref{al1}(d) also shows that a larger transcription rate $s$ will enhance the correlation between the two isoforms. This is consistent with our analytical result in Eq. \eqref{cor2} since a larger value of $s$ results in smaller values of $\beta_1$ and $\beta_2$ and thus results in stronger correlation.

Furthermore, the correlation is also influenced by the regulation strengths $\xi_1$ and $\xi_2$. Fig. \ref{al1}(e),(f) illustrate the correlation coefficient $\rho_{P_1,P_2}$ as a function of $\xi_1$ and $\xi_2$ under different values of $f$, where the two green lines in each figure separate the region with positive correlation (inside the two green lines) and the region with negative correlation (outside the two green lines). One of the two green lines corresponds to the case of $\alpha_1 = 0$ and the other corresponds to the case of $\alpha_2 = 0$. Recall that the two isoforms are positively correlated when $\alpha_1\alpha_2>0$, i.e.
\begin{equation*}
\frac{\eta_2}{\eta_1+f} < \frac{\xi_2}{\xi_1} < \frac{\eta_2+f}{\eta_1}.
\end{equation*}
Therefore, in order to observe positive correlation, $\log\xi_2-\log\xi_1$ must be controlled within a belt-shaped region that becomes wider as $f$ increases (Fig. \ref{al1}(e),(f)). In the absence of regulator ($\xi_1 = \xi_2 = 0$), we have $\alpha_1 = \alpha_2 = 0$ and thus there is no correlation between the two isoforms \cite{wang2014alternative}. If only one of the two isoforms is controlled by the regulator ($\xi_1 > 0$ and $\xi_2 = 0$), we have $\alpha_1<0$ and $\alpha_2 > 0$ and thus they are negatively correlated. If both isoforms are controlled by the regulator ($\xi_1,\xi_2 > 0$), the correlation coefficient can be either positive or negative, depending on whether the degradation rate of nascent mRNA is above or below its critical value.

Finally, we make a crucial observation that the correlation between the two protein isoforms can be either weaker or stronger than that between the two mRNA isoforms, depending on the values of the parameters $L_1$ and $L_2$. Comparing Eq. \eqref{cor1} with Eq. \eqref{cor2}, we can see that the protein correlation is less than the mRNA correlation when $L_1,L_2>1$. However, when $L_1,L_2<1$, i.e. when the translation rates $u_i$ and degradation rates $v_i$ of mRNA isoforms are large compared to the degradation rates $d_i$ of protein isoforms, the protein correlation can be even greater than the mRNA correlation, which means that the translation step may even enhance the correlation between the two isoforms of the gene product.

\section{Discussion}
In this paper, we propose a novel method of computing the joint distribution for a wide class of first-order stochastic reaction networks in steady-state conditions. By allowing all zero-order reactions to occur only at the zero microstate, we simplify the Markovian model of stochastic reaction kinetics to a modified Markovian model whose transition diagram is usually much simpler than that of the original one. In many models of biological relevance, the joint distribution of the modified model can be computed analytically. Finally, the joint generating function of the original model can be recovered from that of the modified model by taking a simple exponential transformation.

While the modified model is generally simpler than the original one, it may not be analytically tractable. However, we show its analytical tractability in two special cases: (i) its irreducible state space is finite and (ii) its irreducible state space has a ladder-shaped topological structure. We provide an easily verifiable criterion for the case (i), which states that if all the first-order reactions except degradation reactions have a conservation law with positive coefficients, then the modified model must have a finite irreducible state space. We also show that the case (ii) is satisfied in many gene expression models of biological interest. Here the ladder-shaped structure results from the fact that for the modified model, we only allow zero-order reactions to occur at the zero microstate. For example, if we allow $\varnothing\rightarrow P$ to occur only at the zero microstate, then the number of $P$ can only vary between $0$ and $1$, which correspond to the two branches of the ladder-shaped structure. In fact, ladder-shaped models have been extensively studied in the literature and their generating functions are always represented by various kinds of hypergeometric functions \cite{melykuti2014equilibrium}. Hence for the case (ii), the generating function of the original model is given by the exponential of hypergeometric functions since an exponential transformation needs to be taken in our approach.

In most previous papers, the exact joint distribution is computed by first converting the CME into a system of PDEs satisfied by the joint generating function and then solving the system of PDEs using the method of characteristics. Compared with this method which often involves tedious computations, our approach greatly reduces the theoretical complexity. We then validate the effectiveness of our method by applying it to four gene expression models of biological significance. The analytical results obtained reveal some interesting biological phenomena: (i) multimodality can be caused by slow synthesis and degradation of some gene product, even when the gene is constitutively expressed; (ii) in the presence of alternative splicing, the numbers of two mRNA/protein isoforms are negatively regulated if one isoform is controlled by the regulator and the other isoform is not; (iii) if both mRNA/protein isoforms are controlled by the regulator, then their abundances can be either positively or negatively correlated, depending on whether the degradation rate of nascent mRNA is above or below its critical value; (iv) the protein isoform correlation may be even greater than the mRNA isoform correlation when the translation rates and degradation rates of mRNA isoforms are large compared to the degradation rates of protein isoforms.

We emphasize that we construct the modified model by allowing all zero-order reactions to occur only at the zero microstate. Hence, in order to apply our method, the reaction system must have at least one zero-order reaction. However, in some biological systems, there may not be a zero-order reaction involved in the system. For example, consider the following gene expression model with promoter switching \cite{peccoud1995markovian}:
\begin{equation*}
G\xrightarrow{a}G^*,\;\;\;G^*\xrightarrow{b}G,\;\;\;
G^*\xrightarrow{\rho}G^*+P,\;\;\;P\xrightarrow{d}\varnothing,
\end{equation*}
where $G$ and $G^*$ denote the inactive and active states of the promoter, respectively, and $P$ denotes the corresponding protein. Note that in this model, while the total number of genes in the two promoter states is constant, the number of genes in the active (inactive) state is not constant. Therefore, the two switching reactions, $G\xrightarrow{}G^*$ and $G^*\xrightarrow{}G$, as well as the synthesis reaction $G^*\xrightarrow{}G^*+P$, are actually first-order reactions and cannot be regarded as zero-order reactions. In this case, there are no zero-order reactions involved in the system and thus our approach can no longer be applied. This is the major limitation of our method. In the presence of promoter switching, it has been shown that the analytical solution of a gene expression model is usually represented by hypergeometric functions \cite{peccoud1995markovian, shahrezaei2008analytical, zhou2012analytical, hornos2005self, grima2012steady, vandecan2013self, kumar2014exact, bokes2015protein, jia2020small}. In our paper, we do not take promoter switching into account and show that the joint distributions for a class of gene expression models can be represented by the exponential of hypergeometric functions. The reason for this discrepancy is that promoter switching is considered for the former but is not considered for the latter.

The current method is aimed to compute the exact solution of the steady-state joint distribution of first-order reaction kinetics. If a system contains higher-order reactions, then the PDEs satisfied by the generating function involve higher-order partial derivatives and hence it is very difficult to solve these PDEs analytically. Current research work aims to develop novel methods of computing the joint distribution of higher-order stochastic reaction kinetics. We anticipate that the method developed in this paper can be combined with various approximate techniques developed recently \cite{lakatos2015multivariate, zhang2016moment, thomas2014phenotypic, cao2018linear} to solve the joint distribution of complex biochemical reaction networks and gene regulatory networks.

\section*{Acknowledgements}
C. J. acknowledges support from the NSAF grant in National Natural Science Foundation of China with grant No. U1930402. D.-Q. Jiang was supported by National Natural Science Foundation of China with grant No. 11871079.

\section*{Appendices}

\subsection*{Appendix A: Relationship between the generating functions of the two models}
Here we uncover the relationship between the generating functions of the original and modified models. Multiplying $x^n=x_1^{n_1}\dots x_N^{n_N}$ on both sides of Eq. \eqref{mastermodified} and then summing over all microstates, we obtain
\begin{equation}\label{tempappendix1}
\begin{split}
\frac{\partial H}{\partial t} =&\; \sum_{n}\left(\sum_{i=1}^N\sum_{j=1}^{r_i}\tilde{q}_{n+e_i-\nu_{ij},n}\pi_{n+e_i-\nu_{ij}}
-\sum_{i=1}^N\sum_{j=1}^{r_i}\tilde{q}_{n,n+\nu_{ij}-e_i}\pi_n\right)x^n\\
&\;+\sum_{n}\left(\sum_{j=1}^{r_0}\tilde{q}_{n-\nu_{0j},n}\pi_{n-\nu_{0j}}\right)x^n
-\sum_{n}\left(\sum_{j=1}^{r_0}\tilde{q}_{n,n+\nu_{0j}}\pi_{n}\right)x^n\\
:=&\;\textrm{I}+\textrm{II}-\textrm{III}.
\end{split}
\end{equation}
Recall that first-order reactions lead to the same transitions for the two models. It then follows from the classical result about first-order reaction systems (see Appendix A.2 in \cite{reis2018general}) that
\begin{equation*}
\textrm{I} =
\sum_{i=1}^N\sum_{j=1}^{r_i} k_{ij}\left(x^{\nu_{ij}}-x_i\right)\frac{\partial H}{\partial x_i}.
\end{equation*}
Moreover, since $\tilde{q}_{n,n+\nu_{0j}}$ is nonzero only when $n = \mathbf{0}$, it is easy to see that
\begin{equation*}
\textrm{II} = \sum_{j=1}^{r_0}k_{0j}\pi_{\mathbf{0}}x^{\nu_{0j}},\;\;\;
\textrm{III} = \sum_{j=1}^{r_0}k_{0j}\pi_{\mathbf{0}}.
\end{equation*}
Inserting the above two equations into Eq. \eqref{tempappendix1} yields
\begin{equation}\label{g}
\frac{\partial H}{\partial t} = \sum_{i=1}^N\sum_{j=1}^{r_i} k_{ij}\left(x^{\nu_{ij}}-x_i\right)\frac{\partial H}{\partial x_i}+\sum_{j=1}^{r_0} k_{0j}\pi_{\mathbf{0}}\left(x^{\nu_{0j}}-1\right).
\end{equation}
We next prove that Eq. \eqref{expression} holds. If both the original and modified models are at the steady state, then it follows from Eq. \eqref{g} that
\begin{equation*}
\begin{aligned}
&\;\sum_{i=1}^N\sum_{j=1}^{r_i} k_{ij}\left(x^{\nu_{ij}}-x_i\right)\frac{\partial}{\partial x_i}\left(e^{\frac{H-1}{\pi_{\mathbf{0}}}}\right)+\sum_{j=1}^{r_0} k_{0j}e^{\frac{H-1}{\pi_{\mathbf{0}}}}\left(x^{\nu_{0j}}-1\right)\\
=&\; \sum_{i=1}^N\sum_{j=1}^{r_i} k_{ij}\left(x^{\nu_{ij}}-x_i\right)e^{\frac{H-1}{\pi_{\mathbf{0}}}}\frac{1}{\pi_{\mathbf{0}}}\frac{\partial H}{\partial x_i}+\sum_{j=1}^{r_0} k_{0j}e^{\frac{H-1}{\pi_{\mathbf{0}}}}\left(x^{\nu_{0j}}-1\right)\\
=&\; e^{\frac{H-1}{\pi_{\mathbf{0}}}}\frac{1}{\pi_{\mathbf{0}}}\left[\sum_{i=1}^N\sum_{j=1}^{r_i} k_{ij}\left(x^{\nu_{ij}}-x_i\right)\frac{\partial H}{\partial x_i}+\sum_{j=1}^{r_0} k_{0j}\pi_{\mathbf{0}}\left(x^{\nu_{0j}}-1\right)\right] = 0.
\end{aligned}
\end{equation*}
Thus we have
\begin{equation}\label{f2}
\sum_{i=1}^N\sum_{j=1}^{r_i} k_{ij}\left(x^{\nu_{ij}}-x_i\right)\frac{\partial}{\partial x_i}\left(e^{\frac{H-1}{\pi_{\mathbf{0}}}}\right)+\sum_{j=1}^{r_0} k_{0j}e^{\frac{H-1}{\pi_{\mathbf{0}}}}\left(x^{\nu_{0j}}-1\right)=0.
\end{equation}
Comparing Eq. \eqref{f2} with Eq. \eqref{f}, we finally conclude that $F = e^{\frac{H-1}{\pi_{\mathbf{0}}}}$ in steady-state conditions.

\subsection*{Appendix B: Finiteness of the irreducible state space of the modified model}
Here we prove the following criterion: if all the first-order reactions except degradation reactions has a conservation law with positive coefficients, then the modified model must have a finite irreducible state space. To prove this criterion, we need the following lemma.
\begin{figure}[!htb]
\centering\includegraphics[width=0.5\columnwidth]{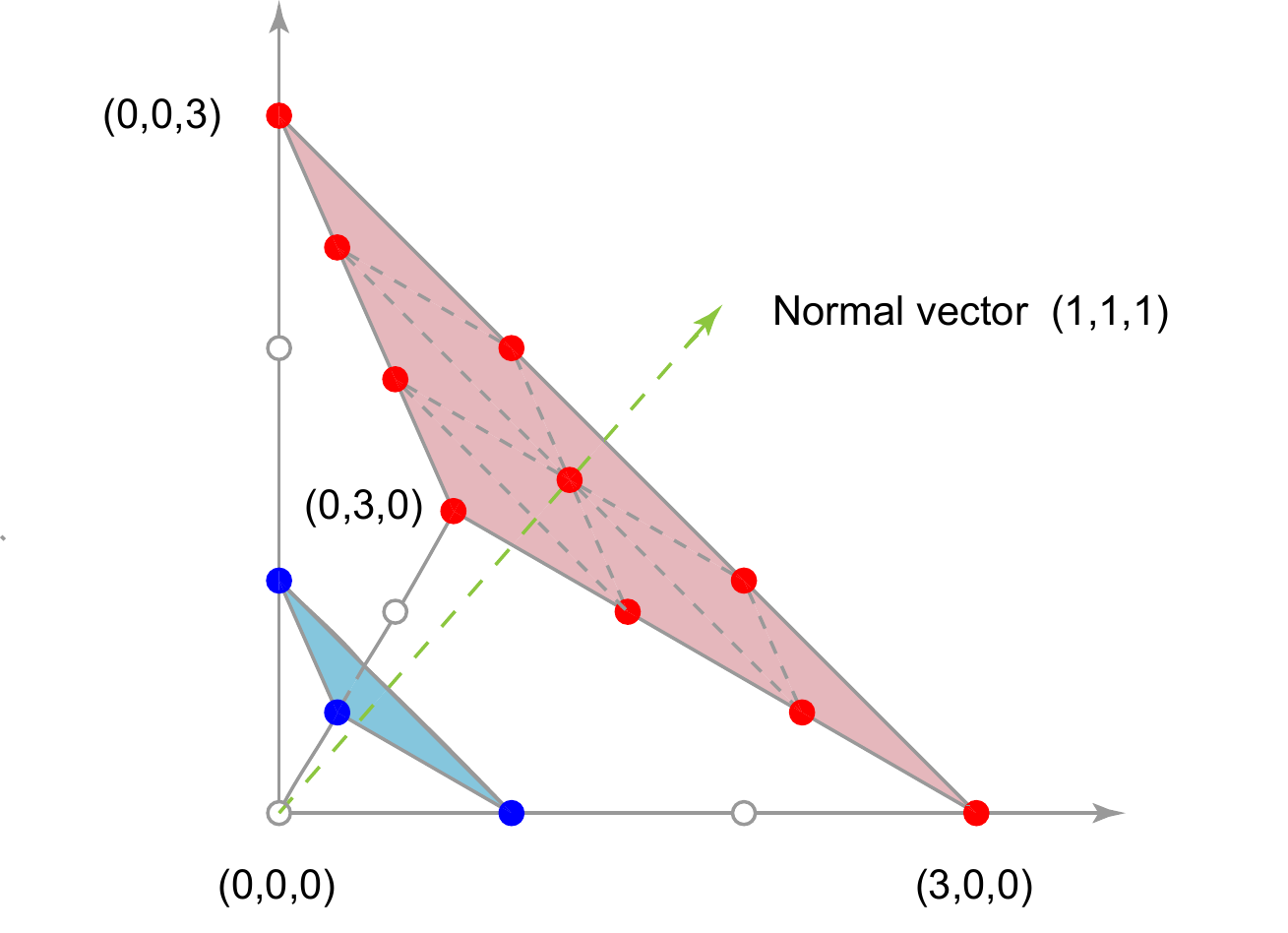}
\caption{Two hyperplanes with the same normal vector (1,1,1). The blue hyperplane contains three points in the first orthant of the lattice space and the red hyperplane contains ten points.}\label{plane}
\end{figure}

\begin{lemma}\label{lemma}
Suppose that a family of reactions
\begin{equation*}
R_i\colon \mu_i^1S_1+\dots+\mu_i^NS_N\xrightarrow{k_i} \nu^1_{i} S_1+\dots+\nu^N_{i} S_N,
\;\;\;i=1,\dots,r,
\end{equation*}
has the conservation law
\begin{equation*}
\omega_1\mu^1_i+\omega_2\mu^2_i+\dots+\omega_N\mu^N_i = \omega_1\nu^1_i+\omega_2\nu^2_i+\dots+\omega_N\nu^N_i,
\end{equation*}
for all $i=1,\dots,r$. If the coefficients $\omega_1,\cdots,\omega_N$ are all positive, then for any microstate $n$, the family of reactions can only lead microstate $n$ to a finite number of microstates.
\end{lemma}

\begin{proof}
For simplicity, we write $\mu_i = (\mu_i^1,\cdots,\mu_i^N)$ and $\nu_i = (\nu_i^1,\cdots,\nu_i^N)$. Suppose that the family of reactions lead microstate $n$ to microstate $\bar{n}$. Then there exists nonnegative integers $\xi_1,\cdots,\xi_r$ such that
\begin{equation*}
\bar{n} = n+\xi_1(\nu_1-\mu_1)+\dots+\xi_r(\nu_r-\mu_r)
\end{equation*}
with $\xi_i$ being the number of occurrence of the $i$th reaction. Then we have
\begin{equation*}
\omega\cdot\bar{n}
= \omega\cdot n+\xi_1\omega \cdot(\nu_1-\mu_1)+\dots+\xi_r\omega\cdot(\nu_r-\mu_r)
= \omega\cdot n,
\end{equation*}
where $\omega\cdot n = \omega_1n_1+\omega_2n_2+\cdots+\omega_Nn_N$ denotes the usual scalar product of vectors. This clearly shows that $\omega\cdot(\bar{n}-n) = 0$, which implies that all the microstates accessible from $n$ must lie in some hyperplane $H$ with normal vector $\omega$. Since the normal vector $\omega$ has positive components, it always points into the first orthant and thus the hyperplane $H$ can only contain a finite number of microstates within the first orthant (see Fig. \ref{plane} for an illustration). This completes the proof.
\end{proof}

We are now in a position to prove the above criterion. Since the original model is ergodic, all nonzero microstates can lead to the zero microstate via a series of first-order reactions. Since first-order reactions result in the same transitions for the original and modified models, for the modified model, all nonzero microstates can also lead to the zero microstate via a series of first-order reactions. This shows that the zero microstate is contained in the irreducible state space of the modified model. Therefore, to identify the irreducible state space of the modified model, we only need to determine which microstates are accessible from the zero microstate. First, since zero-order reactions can only occur at the zero microstate for the modified model, all zero-order reactions can only lead the zero microstate to a finite number of microstates, denoted by $n_1,\cdots,n_k$. Next, since the family of first-order reactions except degradation reactions has a conservation law with positive coefficients, it follows from Lemma \ref{lemma} that all first-order reactions can only lead microstates $n_1,\cdots,n_k$ to a finite number of microstates. This completes the proof of the criterion.

%

\subsection*{Appendix C: Joint distribution for the gene expression model with nascent mRNA}
Let $\pi_{m_\star,m,n}$ denote the steady-state probability of observing microstate $(m_\star,m,n)$ for the modified model. From the transition diagram in Fig. \ref{mrnaall}(b), these steady-state probabilities satisfy the following equations:
\begin{equation}\label{nascentme}
\left\{\begin{split}
&f\pi_{1,0,0}+v\pi_{0,1,0}+d\pi_{0,0,1}-s\pi_{0,0,0}=0,\\
&s\pi_{0,0,0}-(k+f)\pi_{1,0,0}=0,\\
&k\pi_{1,0,0}+d\pi_{0,1,1}-(u+v)\pi_{0,1,0}=0,\\
&v\pi_{0,1,n}+(n+1)d\pi_{0,0,n+1}-nd\pi_{0,0,n}=0,\;\;\;n\geq 1,\\
&u\pi_{0,1,n-1}+(n+1)d\pi_{0,1,n+1}-(nd+u+v)\pi_{0,1,n}=0,\;\;\;n\geq 1.\\
\end{split}\right.
\end{equation}
To proceed, we define the following two generating functions:
\begin{equation*}
\phi(y)=\sum_{n=0}^\infty\pi_{0,0,n}y^n,\;\;\;\psi(y)=\sum_{n=0}^\infty\pi_{0,1,n}y^n.
\end{equation*}
Then the generating function of the modified model is given by
\begin{equation}\label{gmodel1}
H(x_\star,x,y)=\pi_{1,0,0}x_\star+\phi(y)+x\psi(y).
\end{equation}
Note that Eq. \eqref{nascentme} can be converted into the following system of ODEs:
\begin{gather}
\label{eq21} k\pi_{1,0,0}+(uy-u-v)\psi(y)+d(1-y)\psi'(y) = 0,\\
\label{eq22} -k\pi_{1,0,0}+v\psi(y)+d(1-y)\phi'(y)= 0.
\end{gather}
By the second equation in Eq. \eqref{nascentme} we obtain
\begin{equation*}
\pi_{1,0,0}=a\pi_{0,0,0},
\end{equation*}
where $a=s/(k+f)$. Taking the derivative on both sides of Eq. \eqref{eq21} yields
\begin{equation*}
d(1-y)\psi''(y)+\left(uy-u-v-d\right)\psi'(y)+u\psi(y)=0.
\end{equation*}
This is a confluent hypergeometric differential equation \citep[Eq. 13.2.1]{olver2010nist} and its solution is given by
\begin{equation*}
\psi(y) = K{}_1F_1\left(1;1+\nu;\mu(y-1)\right),
\end{equation*}
where $\nu=v/d, \mu=u/d$ and $K$ is a normalization constant. Taking $y=1$ in Eq. \eqref{eq21}, we can determine the normalization constant $K$ as
\begin{equation*}
K = \psi(1) = b\pi_{0,0,0},
\end{equation*}
where $b=ks/(k+f)v$. On the other hand, it follows from Eq. \eqref{eq22} and the power series expansion of the confluent hypergeometric function that
\begin{equation}\label{technique}
\begin{split}
\phi'(y) &= \frac{b\nu\pi_{0,0,0}}{y-1}[{}_1F_1(1;1+\nu;\mu(y-1))-1]\\
&= \frac{b\nu\pi_{0,0,0}}{y-1}\sum_{i=1}^\infty \frac{(\mu(y-1))^i}{(1+\nu)_i}\\
&= \frac{b\nu\mu\pi_{0,0,0}}{1+\nu}\sum_{i=0}^\infty \frac{(\mu(y-1))^i}{(2+\nu)_{i}}\\
&= \frac{b\nu\mu\pi_{0,0,0}}{1+\nu}{}_1F_1\left(1;2+\nu;\mu(y-1)\right),
\end{split}
\end{equation}
where $(x)_i=x(x+1)\dots (x+i-1)$ is the Pochhammer symbol. Thus we obtain
\begin{equation*}
\phi(y)=\frac{b\nu\mu\pi_{0,0,0}}{1+\nu}\int_{1}^y {}_1F_1\left(1;2+\nu;\mu(z-1)\right)dz+C,
\end{equation*}
where $C$ is an undetermined constant. It then follows from Eqs. \eqref{expression} and \eqref{gmodel1} that the generating function of the original model is given by
\begin{equation*}
\begin{split}
F(x_\star,x,y) &= e^{\frac{a\pi_{0,0,0}x_\star+\phi(y)+x\psi(y)-1}{\pi_{0,0,0}}}\\
&= e^{a x_\star+\frac{b\nu\mu}{1+\nu}\int_1^y {}_1F_1(1;2+\nu;\mu(z-1))dz+bx{}_1F_1(1;1+\nu;\mu(y-1))+(C-1)/\pi_{0,0,0}}.
\end{split}
\end{equation*}
By using the fact that $F(1,1,1) = 1$, we can determined the constant $C$ and thus the generating function can be rewritten as
\begin{equation}\label{temp}
F(x_\star,x,y) = e^{a(x_\star-1)+\frac{b\nu\mu}{1+\nu}\int_1^y {}_1F_1(1;2+\nu;\mu(z-1))dz+b[x{}_1F_1(1;1+\nu;\mu(y-1))-1]}.
\end{equation}
To proceed, recall that the confluent hypergeometric function satisfies the following recurrence relation \cite[Eq. 13.3.3]{olver2010nist}:
\begin{equation}\label{tool2}
{}_1F_1\left(2;2+\nu;\mu(z-1)\right)+\nu{}_1F_1\left(1;2+\nu;\mu(z-1)\right)-(1+\nu){}_1F_1\left(1;1+\nu;\mu(z-1)\right)=0.
\end{equation}
Moreover, it follows from the differentiation formula of confluent hypergeometric functions \cite[Eq. 13.3.15]{olver2010nist} that
\begin{equation*}
\frac{d}{dz}{}_1F_1\left(1;1+\nu;\mu(z-1)\right)
= \frac{\mu}{1+\nu}{}_1F_1\left(2;2+\nu;\mu(z-1)\right).
\end{equation*}
Integrating both sides of Eq. \eqref{tool2} from $1$ to $y$, we obtain
\begin{equation}
\begin{aligned}\label{replace}
&\;\frac{\mu\nu}{1+\nu}\int_1^{y}{}_1F_1\left(1;2+\nu;\mu(z-1)\right)dz\\
=&\; \mu\int_1^{y}{}_1F_1\left(1;1+\nu;\mu(z-1)\right)dz-{}_1F_1\left(1;1+\nu;\mu(y-1)\right)+1.
\end{aligned}
\end{equation}
Finally, inserting the above equation into Eq. \eqref{temp}, we obtain Eq. \eqref{preexpression} in the main text.

\subsection*{Appendix D: Joint distribution for the gene regulatory model with translational bursting}
Let $\pi_{n_1,n_2}$ denote the steady-state probability of observing microstate $(n_1,n_2)$ for the modified model. From the transition diagram in Fig. \ref{cascadeall}(b), these steady-state probabilities satisfy the following equations:
\begin{equation}\label{cascadecme}
\left\{\begin{aligned}
&d_1\pi_{1,0}+d_2 \pi_{0,1}-u_1\pi_{0,0}=0,\\
&u_1\pi_{0,0}+d_2 \pi_{1,1}+2d_1\pi_{2,0}-(d_1+u_2p)\pi_{1,0}=0,\\
&d_1\pi_{1,n_2}+(n_2+1)d_2\pi_{0,n_2+1}-n_2d_2 \pi_{0,n_2}=0,\;\;\; n_2\geq 1,\\
&(n_2+1)d_2\pi_{1,n_2+1}+\sum_{i=0}^{n_2-1}u_2p^{n_2-i}q\pi_{1,i}-(u_2p+n_2d_2+d_1)\pi_{1,n_2}=0, \;\;\; n_2\geq 1.
\end{aligned}\right.
\end{equation}
To proceed, we define the following two generating functions:
\begin{equation*}
\phi(y_2)=\sum_{n_2=0}^\infty \pi_{0,n_2}y_2^{n_2},\;\;\; \psi(y_2)=\sum_{n_2=0}^\infty \pi_{1,n_2}y_2^{n_2}.
\end{equation*}
Then the generating function of the modified model can be written as
\begin{equation}\label{gmodel2}
H(y_1,y_2)=\phi(y_2)+y_1\psi(y_2).
\end{equation}
Note that Eq. \eqref{cascadecme} can be converted into the following system of ODEs:
\begin{gather}
\label{cascade1}u_1\pi_{0,0}+\left[\frac{u_2p(y_2-1)}{1-py_2}-d_1\right]\psi(y_2)+d_2(1-y_2)\psi'(y_2)=0.\\
\label{cascade2} -u_1\pi_{0,0}+d_2(1-y_2)\phi'(y_2)+d_1\psi(y_2) = 0.
\end{gather}
Taking the derivative on both sides of Eq. \eqref{cascade1} yields
\begin{equation*}
a(y_2)\psi''(y_2)+b(y_2)\psi'(y_2)+c(y_2)\psi(y_2)=0,
\end{equation*}
where
\begin{gather*}
a(y_2)=(py_2-1)^2(y_2-1),\\
b(y_2)=(py_2-1)[(\mu_2+\nu+1)py_2-(\mu_2p+\nu+1)],\\
c(y_2)=\mu_2p(p-1).
\end{gather*}
This is a hypergeometric differential equation and its solution is given by
\begin{equation*}
\psi(y_2) = K{}_2F_1\left(-\mu_2,1;1+\nu;\omega(y_2)\right),
\end{equation*}
where $\omega(y_2)=p(y_2-1)/(py_2-1)$ and $K$ is a normalization constant. Taking $y_2=1$ in Eq. \eqref{cascade1}, the normalization constant can be determined as
\begin{equation*}
K = \psi(1) = \mu_1\pi_{0,0}.
\end{equation*}
Next we compute $\phi(y_2)$ by using Eq. \eqref{cascade2}. On the other hand, it follows from Eq. \eqref{cascade2} and the power series expansion of the hypergeometric function that
\begin{equation*}
\begin{aligned}
\phi'(y_2) &= \frac{\mu_1\nu\pi_{0,0}}{y_2-1}[{}_2F_1\left(-\mu_2,1;1+\nu;w(y_2)\right)-1]\\
&= \frac{\mu_1\nu p\pi_{0,0}}{(py_2-1)w(y_2)}\sum_{i=1}^\infty\frac{(-\mu_2)_iw(y_2)^i}{(1+\nu)_i}\\
&= \frac{\mu_1\nu p\pi_{0,0}}{(py_2-1)}\sum_{i=0}^\infty\frac{(-\mu_2)_{i+1}w(y_2)^i}{(1+\nu)_{i+1}}\\
&= \frac{-\mu_1\mu_2\nu p\pi_{0,0}}{1+\nu}\frac{{}_2F_1(1-\mu_2,1;2+\nu;w(y_2))}{py_2-1}.
\end{aligned}
\end{equation*}
Thus we obtain
\begin{equation}\label{phi2}
\phi(y_2) = \frac{-\mu_1\mu_2\nu p\pi_{0,0}}{1+\nu}
\int_1^{y_2}\frac{{}_2F_1\left(1-\mu_2,1;2+\nu;w(z)\right)}{pz-1}dz+C,
\end{equation}
where $C$ is an undetermined constant. It then follows from Eqs. \eqref{expression} and \eqref{gmodel2} that the generating function of the original model is given by
\begin{equation}\label{temp2}
\begin{aligned}
F(y_1,y_2) &= e^{\frac{\phi(y_2)+y_1\psi(y_2)-1}{\pi_{0,0}}}\\
&= e^{\mu_1y_1{}_2F_1(-\mu_2,1;1+\nu;\omega(y_2))-\frac{\mu_1\mu_2\nu p}{1+\nu}\int_1^{y_2}
	\frac{{}_2F_1\left(1-\mu_2,1;2+\nu;\omega(z)\right)}{pz-1}dz+(C-1)/\pi_{0,0}}.
\end{aligned}
\end{equation}
By using the fact that $F(1,1) = 1$, we can determined the constant $C$ and thus the generating function can be rewritten as
\begin{equation}\label{temp3}
F(y_1,y_2) = e^{\mu_1[y_1{}_2F_1(-\mu_2,1;1+\nu;\omega(y_2))-1]-\frac{\mu_2\mu_1\nu p}{1+\nu}\int_1^{y_2}
	\frac{{}_2F_1\left(1-\mu_2,1;2+\nu;\omega(z)\right)}{pz-1}dz}.
\end{equation}
To proceed, recall that the hypergeometric function satisfies the following recurrence relation \cite[Eqs. 15.5.13 and 15.5.15]{olver2010nist}:
\begin{equation}\label{cascade10}
\begin{split}
&\nu{}_2F_1\left(1-\mu_2,1;2+\nu;\omega(z)\right)-(1+\nu){}_2F_1\left(-\mu_2,1;1+\nu;\omega(z)\right)\\
&+(1-\omega(z)){}_2F_1\left(1-\mu_2,2;2+\nu;\omega(z)\right)=0.\\
\end{split}
\end{equation}
Since
\begin{equation*}
\omega'(z) = \frac{p(1-\omega(z))}{pz-1},
\end{equation*}
multiplying $\mu_1p/(pz-1)$ on both sides of Eq. \eqref{cascade10} yields
\begin{equation}\label{tool3}
\begin{split}
&\mu_1\nu p\frac{{}_2F_1\left(1-\mu_2,1;2+\nu;\omega(z)\right)}{pz-1}
-(1+\nu)\mu_1p\frac{\;{}_2F_1\left(-\mu_2,1;1+\nu;\omega(z)\right)}{pz-1}\\
&+\mu_1\omega'(z)\;{}_2F_1\left(1-\mu_2,2;2+\nu;\omega(z)\right)=0.\\
\end{split}
\end{equation}
Moreover, it follows from the differentiation formula of Gaussian hypergeometric functions \cite[Eq.
15.5.1]{olver2010nist} that
\begin{equation*}
\frac{d}{dz}{}_2F_1\left(-\mu_2,1;1+\nu;\omega(z)\right)
= \frac{-\mu_2}{1+\nu}\omega'(z){}_2F_1\left(1-\mu_2,2;2+\nu;\omega(z)\right).
\end{equation*}
Integrating both sides of Eq. \eqref{tool3} from $1$ to $y_2$, we obtain
\begin{equation*}
\begin{aligned}
&\; \frac{\mu_1\mu_2\nu p}{1+\nu}\int_1^{y_2}\frac{{}_2F_1\left(1-\mu_2,1;2+\nu;\omega(z)\right)}{pz-1}dz\\
=&\; \mu_1\mu_2p\int_1^{y_2}\frac{{}_2F_1\left(-\mu_2,1;1+\nu;\omega(z)\right)}{pz-1}dz
+\mu_1{}_2F_1\left(-\mu_2,1;1+\nu;\omega(y_2)\right)-\mu_1.
\end{aligned}
\end{equation*}
Inserting the above equation into Eq. \eqref{temp3}, we obtain
\begin{equation*}
F(y_1,y_2) = e^{\mu_1{}_2F_1\left(-\mu_2,1;1+\nu;\omega(y_2)\right)(y_1-1)-\mu_1\mu_2 p\int_1^{y_2}\frac{{}_2F_1\left(-\mu_2,1;1+\nu;\omega(z)\right)}{pz-1}dz}.
\end{equation*}
Finally, using the Kummer's transformation \cite[Eq. 15.5.1]{olver2010nist}, we obtain Eq. \eqref{cascadeexpression} in the main text.

\subsection*{Appendix E: Joint distribution for the gene expression model with alternative splicing}\label{alappendix}
Let $\pi_{m,m_1,n_1,m_2,n_2}$ denote the steady-state probability of observing microstate $(m,m_1,n_1,m_2,n_2)$ for the modified model. Given that there are $n$ copies of the regulator, these steady-state probabilities satisfy the following equations:
\begin{equation}\label{aleq}
\left\{\begin{aligned}
&f\pi_{1,0,0,0,0}+v_1\pi_{0,1,0,0,0}+v_2\pi_{0,0,0,1,0}+d_1\pi_{0,0,1,0,0}+d_2\pi_{0,0,0,0,1}-s\pi_{\mathbf{0}}=0,\\
&s\pi_{0,0,0,0,0}-(k_1(n)+k_2(n)+f)\pi_{1,0,0,0,0}=0,\\
&k_1(n)\pi_{1,0,0,0,0}+d_1\pi_{0,1,1,0,0}-(u_1+v_1)\pi_{0,1,0,0,0}=0,\\
&k_2(n)\pi_{1,0,0,0,0}+d_2\pi_{0,0,0,1,1}-(u_2+v_2)\pi_{0,0,0,1,0}=0,\\
&v_1\pi_{0,1,n_1,0,0}+(n_1+1)d_1\pi_{0,0,n_1+1,0,0}-n_1d_1\pi_{0,0,n_1,0,0}=0,\;\;\;n_1\geq 1,\\
&u_1\pi_{0,1,n_1-1,0,0}+(n_1+1)d_1\pi_{0,0,n_1+1,0,0}-(n_1d_1+u_1+v_1)\pi_{0,1,n_1,0,0}=0,\;\;\;n_1\geq 1,\\
&v_2\pi_{0,0,0,1,n_2}+(n_2+1)d_2\pi_{0,0,0,0,n_2+1}-n_2d_2\pi_{0,0,0,0,n_2}=0,\;\;\;n_2\geq 1,\\
&u_2\pi_{0,0,0,1,n_2-1}+(n_2+1)d_2\pi_{0,0,0,0,n_2+1}-(n_2d_2+u_2+v_2)\pi_{0,0,0,1,n_2}=0,\;\;\;n_2\geq 1.\\
\end{aligned}\right.
\end{equation}
To proceed, we define the following generating functions:
\begin{equation*}
\begin{aligned}
\phi_1(y_1)&=\sum_{n_1=0}^\infty\pi_{0,0,n_1,0,0}y_1^{n_1},\;\;\;\psi_1(y_1)=\sum_{n_1=0}^\infty\pi_{0,1,n_1,0,0}y_1^{n_1},\\ \phi_2(y_2)&=\sum_{n_2=1}^\infty\pi_{0,0,0,0,n_2}y_2^{n_2},\;\;\;\psi_2(y_2)=\sum_{n_2=0}^\infty\pi_{0,0,0,1,n_2}y_2^{n_2}.
\end{aligned}
\end{equation*}
Then, given that there are $n$ copies of the regulator, the generating function of the modified model is given by
\begin{equation}\label{gmodel3}
H(x,x_1,y_1,x_2,y_2|n)=\pi_{1,0,0,0,0}x+\phi_1(y_1)+\psi_1(y_1)x_1+\phi_2(y_2)+\psi_2(y_2)x_2.
\end{equation}
Note that Eq. \eqref{aleq} can be converted into the following system of ODEs:
\begin{gather}
\label{equation51}\pi_{1,0,0,0,0}k_i(n)+(u_iy_i-u_i-v_i)\psi_i(y_i)+d_i(1-y_i)\psi_i'(y_i)=0,\\
\label{equation52}-\pi_{1,0,0,0,0}k_i(n)+v_i\psi_i(y_i)+d_i(1-y_i)\phi_i'(y_i)=0,
\end{gather}
for $i=1,2$. By the second equation in Eq. \eqref{aleq} we obtain
\begin{equation*}
\pi_{1,0,0,0,0} = a(n)\pi_{\mathbf{0}},
\end{equation*}
where $a(n)=s/(k_1(n)+k_2(n)+f)$. Note that Eqs. \eqref{equation51} and \eqref{equation52} have a similar form as Eqs. \eqref{eq21} and \eqref{eq22}. By using the same procedure used for solving Eqs. \eqref{eq21} and \eqref{eq22}, we obtain
\begin{gather*}
\psi_i(y_i) = K_i(n)b_i\pi_{\mathbf{0}}{}_1F_1\left(1;1+\nu_i;\mu(y_i-1)\right),\\
\phi_i(y_i) = \frac{K_i(n)b_i\nu_i\mu_i\pi_{\mathbf{0}}}{1+\nu_i}
\int_{1}^{y_i}{}_1F_1\left(1;2+\nu_i;\mu_i(z-1)\right)dz+C_i,
\end{gather*}
where $C_i$ are two undetermined constants and
\begin{gather*}
K_1(n)=\frac{k_1(n)}{k_1(n)+k_2(n)+f},\;\;\;K_2(n)=\frac{k_2(n)}{k_1(n)+k_2(n)+f},\\
b_1=\frac{s}{v_1},\;\;\;b_2=\frac{s}{v_2},\;\;\;\mu_1=\frac{u_1}{d_1},\;\;\;\mu_2=\frac{u_2}{d_2},\;\;\;\nu_1=\frac{v_1}{d_1},\;\;\;\nu_2=\frac{v_2}{d_2}.
\end{gather*}
It thus follows from Eqs. \eqref{expression} and \eqref{gmodel3} that the generating function of the original model, given that there are $n$ copies of the regulator, is given by
\begin{equation*}
F(x,x_1,y_1,x_2,y_2|n)=e^{a(n)(x-1)+\sum_{i=1}^2K_i(n)b_i\left[x_i{}_1F_1(1;1+\nu_i;\mu_i(y_i-1))
	+\frac{\mu_i\nu_i}{1+\nu_i}\int_1^{y_i}{}_1F_1(1;2+\nu_i;\mu_i(z-1))dz\right]}.
\end{equation*}
Replacing $\mu$, $\nu$, and $y$ in Eq. \eqref{replace} by $\mu_i$, $\nu_i$, and $y_i$ for $i = 1,2$ and inserting the resulting two equations into the above equation give Eq. \eqref{final4} in the main text.

Next we compute the correlation coefficients between the copy numbers of the two mRNA/protein isoforms under the assumption that the copy number of the regulator has a Poisson distribution with parameter $\lambda$. In this case, the generating function of the original model is given by
\begin{equation*}
\begin{aligned}
&F(x,x_1,y_1,x_2,y_2)\\
&=\sum_{n=0}^\infty\frac{\lambda^ne^{-\lambda}}{n!} e^{a(n)(x-1)+\sum_{i=1}^2K_i(n)b_i\left[(x_i-1){}_1F_1(1;1+\nu_i;\mu_i(y_i-1))+\mu_i\int_1^{y_i}{}_1F_1(1;1+\nu_i;\mu_i(z-1))dz\right]}.
\end{aligned}
\end{equation*}
We first focus on the correlation between the two mRNA isoforms. Using the power series expansion and the Kummer transformation \cite[Eq. 13.2.39]{olver2010nist} of confluent hypergeometric functions, the derivative of $F$ with respect to $x_i$ is given by
\begin{equation}\label{technique1}
\begin{aligned}
\frac{\partial F}{\partial x_i}\left(1,1,1,1,1\right)&=\sum_{n=0}^\infty K_i(n)b_i\frac{\lambda^ne^{-\lambda}}{n!}=\frac{b_i}{\xi_1+\xi_2}\sum_{n=0}^\infty \left(\xi_i+\frac{\alpha_i\gamma}{n+\gamma}\right)\frac{\lambda^{n}e^{-\lambda}}{n!}\\
&=\frac{b_i}{\xi_1+\xi_2}\sum_{n=0}^\infty \left(\xi_i+\frac{\alpha_i(\gamma)_n}{(\gamma+1)_n}\right)\frac{\lambda^{n}e^{-\lambda}}{n!}
=\frac{b_i\left(\xi_i+\alpha_ih_1\right)}{\xi_1+\xi_2},
\end{aligned}
\end{equation}
and the second derivative of $F$ with respect to $x_i$ and $x_j$ is given by
\begin{equation}\label{technique2}
\begin{aligned}
\frac{\partial^2 F}{\partial x_i \partial x_j}(1,1,1,1,1)& =\sum_{n=0}^\infty K_i(n)K_j(n)b_ib_j\frac{\lambda^{n}e^{-\lambda}}{n!}\\
&=\frac{b_ib_j}{\left(\xi_1+\xi_2\right)^2}\sum_{n=0}^\infty \left[\xi_i\xi_j+\frac{(\alpha_i\xi_j+\alpha_j\xi_i)(\gamma)_n}{(\gamma+1)_n}
+\frac{\alpha_i\alpha_j(\gamma)_n(\gamma)_n}{(\gamma+1)_n(\gamma+1)_n}\right]\frac{\lambda^{n}e^{-\lambda}}{n!}\\
&=\frac{b_ib_j\left[\xi_i\xi_j+\left(\alpha_i\xi_j+\alpha_j\xi_i\right)h_1+\alpha_i\alpha_jh_2\right]}
{\left(\xi_1+\xi_2\right)^2},
\end{aligned}
\end{equation}
where
\begin{gather*}
\alpha_1 = \frac{\xi_2\eta_1-\xi_1\eta_2-\xi_1f}{\eta_1+\eta_2+f},\;\;\;
\alpha_2 = \frac{\xi_1\eta_2-\xi_2\eta_1-\xi_2f}{\eta_1+\eta_2+f},\;\;\;
\gamma = \frac{\eta_1+\eta_2+f}{\xi_1+\xi_2},\\
h_1 = {}_1F_1\left(1;\gamma+1;-\lambda\right),\;\;\;
h_2 = {}_2F_2\left(\gamma,\gamma;\gamma+1,\gamma+1;\lambda\right)e^{-\lambda}.
\end{gather*}
Inserting the above two equations into Eq. \eqref{correlation}, we obtain Eq. \eqref{cor1} in the main text. We next focus on the correlation between the two protein isoforms. Using the power series expansion
and the Kummer transformation \cite[Eq. 13.2.39]{olver2010nist} of confluent hypergeometric functions, it is not hard to prove that
\begin{equation*}
\begin{aligned}
\frac{\partial F}{\partial y_i}(1,1,1,1,1)
&=\sum_{n=0}^\infty K_i(n)\mu_ib_i\frac{\lambda^ne^{-\lambda}}{n!}
=\frac{\mu_ib_i\left(\xi_i+\alpha_ih_1\right)}{\xi_1+\xi_2}.
\end{aligned}
\end{equation*}
Similarly, the second derivative of $F$ with respect to $y_1$ and $y_2$ is given by
\begin{equation*}
\begin{aligned}
\frac{\partial^2 F}{\partial y_1\partial y_2}(1,1,1,1,1)
&=\sum_{n=0}^\infty K_1(n)K_2(n)\mu_1\mu_2b_1b_2\frac{\lambda^ne^{-\lambda}}{n!}\\
&=\frac{\mu_1\mu_2b_1b_2\left[\xi_1\xi_2+\left(\alpha_1\xi_2+\alpha_2\xi_1\right)h_1
	+\alpha_1\alpha_2h_2\right]}{(\xi_1+\xi_2)^2},
\end{aligned}
\end{equation*}
and the second derivative of $F$ with respect to $y_i$ is given by
\begin{equation*}
\begin{aligned}
\frac{\partial^2 F}{\partial y_i^2}(1,1,1,1,1)
&=\sum_{n=0}^\infty \left[\left(K_i(n)\mu_ib_i\right)^2+K_i(n)\mu_ib_i\frac{\mu_i}{1+\nu_i}\right]
\frac{\lambda^ne^{-\lambda}}{n!}\\
&=\frac{\mu^2_ib_i^2\left(\xi_i^2+2\alpha_i\xi_ih_1+\alpha_i^2h_2\right)}
{(\xi_1+\xi_2)^2}+\frac{\mu^2_ib_i\left(\xi_i+\alpha_ih_1\right)}{(\xi_1+\xi_2)(1+\nu_i)}.
\end{aligned}
\end{equation*}
Inserting the above three equations into Eq. \eqref{correlation} gives Eq. \eqref{cor2} in the main text.

Finally we prove that $\xi_i+\alpha_ih_1$ and $h_2-h_1^2$ are positive for any choice of rate constants. First, we note that
\begin{gather*}
\frac{\alpha_i}{\xi_i} = \frac{\eta_i(\xi_1+\xi_2)}{\xi_i(\eta_1+\eta_2+f)}-1 > -1,\\
h_1 = e^{-\lambda}{}_1F_1\left(\gamma;\gamma+1;\lambda\right)
= e^{-\lambda}\sum_{n=0}^\infty \frac{\left(\gamma\right)_n}{\left(\gamma+1\right)_{n}}
\frac{\lambda^n}{n!} < 1.
\end{gather*}
Combining the above inequalities shows that $\xi_i+\alpha_ih_1>0$. Second, it follows from the Cauchy product formula of two infinite series that
\begin{equation}\label{cauchy}
\begin{aligned}
h_2-h_1^2&=e^{-2\lambda}\left[e^{\lambda}{}_{2}F_2(\gamma,\gamma;\gamma+1,\gamma+1;\lambda)
-\left({}_1F_1(\gamma;\gamma+1;\lambda)\right)^2\right]\\
&=e^{-2\lambda}\left[\sum_{n=0}^\infty \frac{\lambda^n}{n!}\sum_{n=0}^\infty\left(\frac{\gamma}{\gamma+n}\right)^2 \frac{\lambda^n}{n!}-\left(\sum_{n=0}^\infty \frac{\gamma}{\gamma+n}\frac{\lambda^n}{n!}\right)^2\right]\\
&=e^{-2\lambda}\gamma^2\left[\sum_{n=0}^\infty\frac{\lambda^n}{n!}
\sum_{i=0}^n\binom{i}{n}\left[\frac{1}{(\gamma+i)^2}-\frac{1}{(\gamma+i)(\gamma+n-i)}\right]\right].
\end{aligned}
\end{equation}
We next prove that
\begin{equation}\label{proof}
\sum_{i=0}^n\binom{i}{n}\left[\frac{1}{(\gamma+i)^2}-\frac{1}{(\gamma+i)(\gamma+n-i)}\right] > 0,
\end{equation}
for any $\gamma>0$ and $n> 0$. Putting the first term and the last term in the left-hand size of Eq. \eqref{proof} together yields
\begin{equation*}
\begin{aligned}
&\;\left[\frac{1}{\gamma^2}-\frac{1}{\gamma(\gamma+n)}\right]
+\left[\frac{1}{(\gamma+n)^2}-\frac{1}{\gamma(\gamma+n)}\right]\\
=&\;\frac{n}{\gamma^2(\gamma+n)}-\frac{n}{(\gamma+n)^2\gamma}=\frac{n^2}{\gamma^2(\gamma+n)^2}> 0.
\end{aligned}
\end{equation*}
Similarly, putting the second term and the last but one term together gives
\begin{equation*}
\begin{aligned}
&\;n\left[\frac{1}{(\gamma+1)^2}-\frac{1}{(\gamma+1)(\gamma+n-1)}\right]
+n\left[\frac{1}{(\gamma+n-1)^2}-\frac{1}{(\gamma+1)(\gamma+n-1)}\right]\\
=&\;n\left[\frac{n-2}{(\gamma+1)^2(\gamma+n-1)}-\frac{n-2}{(\gamma+n-1)^2(\gamma+1)}\right]
=\frac{n(n-2)^2}{(\gamma+1)^2(\gamma+n-1)^2} > 0.
\end{aligned}
\end{equation*}
If $n$ is an odd number, then repeating the above procedure shows that the left-hand size of Eq. \eqref{proof} is positive. If $n$ is an even number, then the $(n/2+1)$th term in the left-hand size of Eq. \eqref{proof} cannot be paired in the above manner. However, in this case it is easy to check the $(n/2+1)$th term must equal zero. Thus we have proved Eq. \eqref{proof}. Combining Eqs. \eqref{cauchy} and \eqref{proof} finally shows that $h_2-h_1^2>0$.

\setlength{\bibsep}{5pt}
\small\bibliographystyle{nature}

\begin{thebibliography}{0}
\expandafter\ifx\csname natexlab\endcsname\relax\def\natexlab#1{#1}\fi
\expandafter\ifx\csname url\endcsname\relax
  \def\url#1{\texttt{#1}}\fi
\expandafter\ifx\csname urlprefix\endcsname\relax\def\urlprefix{URL }\fi

\end{thebibliography}


\begin{thebibliography}{60}
\expandafter\ifx\csname natexlab\endcsname\relax\def\natexlab#1{#1}\fi
\expandafter\ifx\csname url\endcsname\relax
  \def\url#1{\texttt{#1}}\fi
\expandafter\ifx\csname urlprefix\endcsname\relax\def\urlprefix{URL }\fi

\bibitem[{Anderson \& Kurtz(2015)}]{anderson2015stochastic}
Anderson, D.~F. \& Kurtz, T.~G.
\newblock \emph{{Stochastic Analysis of Biochemical Systems}} (Springer, 2015).

\bibitem[{Qian \& Elson(2002)}]{qian2002single}
Qian, H. \& Elson, E.~L.
\newblock Single-molecule enzymology: stochastic Michaelis--Menten kinetics.
\newblock \emph{Biophys. Chem.} \textbf{101}, 565--576 (2002).

\bibitem[{Jia \emph{et~al.}(2012)Jia, Liu, Qian, Jiang \&
  Zhang}]{jia2012kinetic}
Jia, C., Liu, X.-F., Qian, M.-P., Jiang, D.-Q. \& Zhang, Y.-P.
\newblock Kinetic behavior of the general modifier mechanism of Botts and
  Morales with non-equilibrium binding.
\newblock \emph{J. Theor. Biol.} \textbf{296}, 13--20 (2012).

\bibitem[{Paulsson(2005)}]{paulsson2005models}
Paulsson, J.
\newblock Models of stochastic gene expression.
\newblock \emph{Phys. Life Rev.} \textbf{2}, 157--175 (2005).

\bibitem[{Schnoerr \emph{et~al.}(2014)Schnoerr, Sanguinetti \&
  Grima}]{schnoerr2014complex}
Schnoerr, D., Sanguinetti, G. \& Grima, R.
\newblock The complex chemical Langevin equation.
\newblock \emph{J. Chem. Phys.} \textbf{141}, 07B606\_1 (2014).

\bibitem[{Holehouse \emph{et~al.}(2020)Holehouse, Sukys \&
  Grima}]{holehouse2020stochastic}
Holehouse, J., Sukys, A. \& Grima, R.
\newblock Stochastic time-dependent enzyme kinetics: closed-form solution and
  transient bimodality.
\newblock \emph{J. Chem. Phys.} \textbf{153}, 164113 (2020).

\bibitem[{Peccoud \& Ycart(1995)}]{peccoud1995markovian}
Peccoud, J. \& Ycart, B.
\newblock Markovian modeling of gene-product synthesis.
\newblock \emph{Theor. Popul. Biol.} \textbf{48}, 222--234 (1995).

\bibitem[{Shahrezaei \& Swain(2008)}]{shahrezaei2008analytical}
Shahrezaei, V. \& Swain, P.~S.
\newblock Analytical distributions for stochastic gene expression.
\newblock \emph{Proc. Natl. Acad. Sci. USA} \textbf{105}, 17256--17261 (2008).

\bibitem[{Zhou \& Zhang(2012)}]{zhou2012analytical}
Zhou, T. \& Zhang, J.
\newblock Analytical results for a multistate gene model.
\newblock \emph{SIAM J. Appl. Math.} \textbf{72}, 789--818 (2012).

\bibitem[{Hornos \emph{et~al.}(2005)}]{hornos2005self}
Hornos, J. \emph{et~al.}
\newblock Self-regulating gene: an exact solution.
\newblock \emph{Phys. Rev. E} \textbf{72}, 051907 (2005).

\bibitem[{Grima \emph{et~al.}(2012)Grima, Schmidt \& Newman}]{grima2012steady}
Grima, R., Schmidt, D. \& Newman, T.
\newblock Steady-state fluctuations of a genetic feedback loop: An exact
  solution.
\newblock \emph{J. Chem. Phys.} \textbf{137}, 035104 (2012).

\bibitem[{Vandecan \& Blossey(2013)}]{vandecan2013self}
Vandecan, Y. \& Blossey, R.
\newblock Self-regulatory gene: an exact solution for the gene gate model.
\newblock \emph{Phys. Rev. E} \textbf{87}, 042705 (2013).

\bibitem[{Kumar \emph{et~al.}(2014)Kumar, Platini \& Kulkarni}]{kumar2014exact}
Kumar, N., Platini, T. \& Kulkarni, R.~V.
\newblock Exact distributions for stochastic gene expression models with
  bursting and feedback.
\newblock \emph{Phys. Rev. Lett.} \textbf{113}, 268105 (2014).

\bibitem[{Bokes \& Singh(2015)}]{bokes2015protein}
Bokes, P. \& Singh, A.
\newblock Protein copy number distributions for a self-regulating gene in the
  presence of decoy binding sites.
\newblock \emph{PloS one} \textbf{10}, e0120555 (2015).

\bibitem[{Jia \& Grima(2020{\natexlab{a}})}]{jia2020small}
Jia, C. \& Grima, R.
\newblock Small protein number effects in stochastic models of autoregulated
  bursty gene expression.
\newblock \emph{J. Chem. Phys.} \textbf{152}, 084115 (2020{\natexlab{a}}).

\bibitem[{Jia \& Grima(2020{\natexlab{b}})}]{jia2020dynamical}
Jia, C. \& Grima, R.
\newblock Dynamical phase diagram of an auto-regulating gene in fast switching
  conditions.
\newblock \emph{J. Chem. Phys.} \textbf{152}, 174110 (2020{\natexlab{b}}).

\bibitem[{M{\'e}lyk{\'u}ti \emph{et~al.}(2014)M{\'e}lyk{\'u}ti, Hespanha \&
  Khammash}]{melykuti2014equilibrium}
M{\'e}lyk{\'u}ti, B., Hespanha, J.~P. \& Khammash, M.
\newblock Equilibrium distributions of simple biochemical reaction systems for
  time-scale separation in stochastic reaction networks.
\newblock \emph{J. R. Soc. Interface} \textbf{11}, 20140054 (2014).

\bibitem[{Lakatos \emph{et~al.}(2015)Lakatos, Ale, Kirk \&
  Stumpf}]{lakatos2015multivariate}
Lakatos, E., Ale, A., Kirk, P.~D. \& Stumpf, M.~P.
\newblock Multivariate moment closure techniques for stochastic kinetic models.
\newblock \emph{J. Chem. Phys.} \textbf{143}, 094107 (2015).

\bibitem[{Zhang \emph{et~al.}(2016)Zhang, Nie \& Zhou}]{zhang2016moment}
Zhang, J., Nie, Q. \& Zhou, T.
\newblock A moment-convergence method for stochastic analysis of biochemical
  reaction networks.
\newblock \emph{J. Chem. Phys.} \textbf{144}, 194109 (2016).

\bibitem[{Thomas \emph{et~al.}(2014)Thomas, Popovic \&
  Grima}]{thomas2014phenotypic}
Thomas, P., Popovic, N. \& Grima, R.
\newblock Phenotypic switching in gene regulatory networks.
\newblock \emph{Proc. Natl. Acad. Sci. USA} \textbf{111}, 6994--6999 (2014).

\bibitem[{Cao \& Grima(2018)}]{cao2018linear}
Cao, Z. \& Grima, R.
\newblock Linear mapping approximation of gene regulatory networks with
  stochastic dynamics.
\newblock \emph{Nat. Commun.} \textbf{9}, 1--15 (2018).

\bibitem[{Krieger \& Gans(1960)}]{krieger1960first}
Krieger, I.~M. \& Gans, P.~J.
\newblock First-order stochastic processes.
\newblock \emph{J. Chem. Phys.} \textbf{32}, 247--250 (1960).

\bibitem[{Darvey \& Staff(1966)}]{darvey1966stochastic}
Darvey, I. \& Staff, P.
\newblock Stochastic approach to first-order chemical reaction kinetics.
\newblock \emph{J. Chem. Phys.} \textbf{44}, 990--997 (1966).

\bibitem[{Van~Kampen(1976)}]{van1976equilibrium}
Van~Kampen, N.~G.
\newblock The equilibrium distribution of a chemical mixture.
\newblock \emph{Phys. Lett. A} \textbf{59}, 333--334 (1976).

\bibitem[{Gans(1960)}]{gans1960open}
Gans, P.~J.
\newblock Open First-Order Stochastic Processes.
\newblock \emph{J. Chem. Phys.} \textbf{33}, 691--694 (1960).

\bibitem[{Gadgil \emph{et~al.}(2005)Gadgil, Lee \&
  Othmer}]{gadgil2005stochastic}
Gadgil, C., Lee, C.~H. \& Othmer, H.~G.
\newblock A stochastic analysis of first-order reaction networks.
\newblock \emph{Bull. Math. Biol.} \textbf{67}, 901--946 (2005).

\bibitem[{Heuett \& Qian(2006)}]{heuett2006grand}
Heuett, W.~J. \& Qian, H.
\newblock Grand canonical Markov model: a stochastic theory for open
  nonequilibrium biochemical networks.
\newblock \emph{J. Chem. Phys.} \textbf{124}, 044110 (2006).

\bibitem[{Jahnke \& Huisinga(2007)}]{jahnke2007solving}
Jahnke, T. \& Huisinga, W.
\newblock Solving the chemical master equation for monomolecular reaction
  systems analytically.
\newblock \emph{J. Math. Biol.} \textbf{54}, 1--26 (2007).

\bibitem[{Horn \& Jackson(1972)}]{horn1972general}
Horn, F. J.~M. \& Jackson, R.
\newblock General mass action kinetics.
\newblock \emph{Arch. Ration. Mech. An.} \textbf{47}, 81--116 (1972).

\bibitem[{Anderson \emph{et~al.}(2010)Anderson, Craciun \&
  Kurtz}]{anderson2010product}
Anderson, D.~F., Craciun, G. \& Kurtz, T.~G.
\newblock Product-form stationary distributions for deficiency zero chemical
  reaction networks.
\newblock \emph{Bull. Math. Biol.} \textbf{72}, 1947--1970 (2010).

\bibitem[{Cappelletti \& Wiuf(2016)}]{cappelletti2016product}
Cappelletti, D. \& Wiuf, C.
\newblock Product-form poisson-like distributions and complex balanced reaction
  systems.
\newblock \emph{SIAM J. Appl. Math.} \textbf{76}, 411--432 (2016).

\bibitem[{Reis \emph{et~al.}(2018)Reis, Kromer \& Klipp}]{reis2018general}
Reis, M., Kromer, J.~A. \& Klipp, E.
\newblock General solution of the chemical master equation and modality of
  marginal distributions for hierarchic first-order reaction networks.
\newblock \emph{J. Math. Biol.} \textbf{77}, 377--419 (2018).

\bibitem[{Bokes \emph{et~al.}(2012)Bokes, King, Wood \& Loose}]{bokes2012exact}
Bokes, P., King, J.~R., Wood, A.~T. \& Loose, M.
\newblock Exact and approximate distributions of protein and mRNA levels in the
  low-copy regime of gene expression.
\newblock \emph{J. Math. Biol.} \textbf{64}, 829--854 (2012).

\bibitem[{Pendar \emph{et~al.}(2013)Pendar, Platini \&
  Kulkarni}]{pendar2013exact}
Pendar, H., Platini, T. \& Kulkarni, R.~V.
\newblock Exact protein distributions for stochastic models of gene expression
  using partitioning of Poisson processes.
\newblock \emph{Phys. Rev. E} \textbf{87}, 042720 (2013).

\bibitem[{Wang \& Zhou(2014)}]{wang2014alternative}
Wang, Q. \& Zhou, T.
\newblock Alternative-splicing-mediated gene expression.
\newblock \emph{Phys. Rev. E} \textbf{89}, 012713 (2014).

\bibitem[{Norris \emph{et~al.}(1998)Norris, Norris \&
  Norris}]{norris1998markov}
Norris, J.~R., Norris, J.~R. \& Norris, J.~R.
\newblock \emph{Markov chains}.
\newblock No.~2 (Cambridge university press, 1998).

\bibitem[{Jia(2016)}]{jia2016model}
Jia, C.
\newblock Model simplification and loss of irreversibility.
\newblock \emph{Phys. Rev. E} \textbf{93}, 052149 (2016).

\bibitem[{Johnson(2002)}]{johnson2002curious}
Johnson, W.~P.
\newblock The curious history of Fa{\`a} di Bruno's formula.
\newblock \emph{The American mathematical monthly} \textbf{109}, 217--234
  (2002).

\bibitem[{Bell(1927)}]{bell1927partition}
Bell, E.~T.
\newblock Partition polynomials.
\newblock \emph{Annals of Mathematics} 38--46 (1927).

\bibitem[{Ryan \emph{et~al.}(1991)Ryan, King \& Thomas}]{1991Cleavage}
Ryan, M.~D., King, A. M.~Q. \& Thomas, G.~P.
\newblock Cleavage of foot-and-mouth disease virus polyprotein is mediated by
  residues located within a 19 amino acid sequence.
\newblock \emph{J. Gen. Virol.} \textbf{72 ( Pt 11)}, 2727 (1991).

\bibitem[{Andrea \emph{et~al.}(2005)}]{Andrea2005Development}
Andrea \emph{et~al.}
\newblock Development of 2A peptide-based strategies in the design of
  multicistronic vectors: Expert Opinion on Biological Therapy: Vol 5, No 5.
\newblock \emph{Expert Opin. Biol. Ther.}  (2005).

\bibitem[{Liu \emph{et~al.}(2017)}]{2017Systematic}
Liu, Z. \emph{et~al.}
\newblock Systematic comparison of 2A peptides for cloning multi-genes in a
  polycistronic vector.
\newblock \emph{Sci. Rep.} \textbf{7}, 2193 (2017).

\bibitem[{Ryan \emph{et~al.}(2001)}]{2001Analysis}
Ryan, M.~D. \emph{et~al.}
\newblock Analysis of the aphthovirus 2A/2B polyprotein 'cleavage' mechanism
  indicates not a proteolytic reaction, but a novel translational effect: a
  putative ribosomal 'skip'.
\newblock \emph{J. Gen. Virol.} \textbf{82}, 1013--1025 (2001).

\bibitem[{Donnelly \emph{et~al.}(2001)Donnelly, Hughes, Luke, Mendoza \&
  Ryan}]{2001The}
Donnelly, M. L.~L., Hughes, L.~E., Luke, G., Mendoza, H. \& Ryan, M.~D.
\newblock The 'cleavage' activities of foot-and-mouth disease virus 2A
  site-directed mutants and naturally occurring '2A-like' sequences.
\newblock \emph{J. Gen. Virol.} \textbf{82}, 1027 (2001).

\bibitem[{Loukas \& Kemp(1986)}]{loukas1986index}
Loukas, S. \& Kemp, C.
\newblock The index of dispersion test for the bivariate Poisson distribution.
\newblock \emph{Biometrics} 941--948 (1986).

\bibitem[{Munsky \& Khammash(2006)}]{munsky2006finite}
Munsky, B. \& Khammash, M.
\newblock The finite state projection algorithm for the solution of the
  chemical master equation.
\newblock \emph{J. Chem. Phys.} \textbf{124}, 044104 (2006).

\bibitem[{Saitou(2013)}]{saitou2013introduction}
Saitou, N.
\newblock Introduction to evolutionary genomics.
\newblock \emph{J. Math. Biol.}  (2013).

\bibitem[{La~Manno \emph{et~al.}(2018)}]{la2018rna}
La~Manno, G. \emph{et~al.}
\newblock RNA velocity of single cells.
\newblock \emph{Nature} \textbf{560}, 494--498 (2018).

\bibitem[{Li \emph{et~al.}(2020)Li, Shi, Wu \& Zhou}]{li2020mathematics}
Li, T., Shi, J., Wu, Y. \& Zhou, P.
\newblock On the Mathematics of RNA Velocity I: Theoretical Analysis.
\newblock \emph{bioRxiv}  (2020).

\bibitem[{Jia \& Grima(2021)}]{jia2021frequency}
Jia, C. \& Grima, R.
\newblock Frequency domain analysis of fluctuations of mRNA and protein copy
  numbers within a cell lineage: theory and experimental validation.
\newblock \emph{Phys. Rev. X} \textbf{11}, 021032 (2021).

\bibitem[{Cai \emph{et~al.}(2006)Cai, Friedman \& Xie}]{cai2006stochastic}
Cai, L., Friedman, N. \& Xie, X.~S.
\newblock Stochastic protein expression in individual cells at the single
  molecule level.
\newblock \emph{Nature} \textbf{440}, 358--362 (2006).

\bibitem[{Jia(2017)}]{jia2017simplification}
Jia, C.
\newblock Simplification of Markov chains with infinite state space and the
  mathematical theory of random gene expression bursts.
\newblock \emph{Phys. Rev. E} \textbf{96}, 032402 (2017).

\bibitem[{Ajith \emph{et~al.}(2016)}]{ajith2016position-dependent}
Ajith, S. \emph{et~al.}
\newblock Position-dependent activity of CELF2 in the regulation of splicing
  and implications for signal-responsive regulation in T cells.
\newblock \emph{RNA Biol.} \textbf{13}, 569--581 (2016).

\bibitem[{Fu \& Ares(2014)}]{Fu2014Context}
Fu, X.~D. \& Ares, M.
\newblock Context-dependent control of alternative splicing by RNA-binding
  proteins.
\newblock \emph{Nat. Rev. Genet.} \textbf{15}, 689--701 (2014).

\bibitem[{Baralle \& Giudice(2017)}]{2017Alternative}
Baralle, F.~E. \& Giudice, J.
\newblock Alternative splicing as a regulator of development and tissue
  identity.
\newblock \emph{Nat. Rev. Mol. Cell Biol.} \textbf{18} (2017).

\bibitem[{Jia \emph{et~al.}(2017{\natexlab{a}})Jia, Zhang \&
  Qian}]{jia2017emergent}
Jia, C., Zhang, M.~Q. \& Qian, H.
\newblock Emergent L{\'e}vy behavior in single-cell stochastic gene expression.
\newblock \emph{Phys. Rev. E} \textbf{96}, 040402 (2017{\natexlab{a}}).

\bibitem[{Jia \emph{et~al.}(2017{\natexlab{b}})Jia, Xie, Chen \&
  Zhang}]{jia2017stochastic}
Jia, C., Xie, P., Chen, M. \& Zhang, M.~Q.
\newblock Stochastic fluctuations can reveal the feedback signs of gene
  regulatory networks at the single-molecule level.
\newblock \emph{Sci. Rep.} \textbf{7}, 1--9 (2017{\natexlab{b}}).

\bibitem[{Jia \emph{et~al.}(2019)Jia, Yin, Zhang \emph{et~al.}}]{jia2019single}
Jia, C., Yin, G.~G., Zhang, M.~Q. \emph{et~al.}
\newblock Single-cell stochastic gene expression kinetics with coupled
  positive-plus-negative feedback.
\newblock \emph{Phys. Rev. E} \textbf{100}, 052406 (2019).

\bibitem[{Jia \emph{et~al.}(2018)Jia, Qian, Chen \& Zhang}]{jia2018relaxation}
Jia, C., Qian, H., Chen, M. \& Zhang, M.~Q.
\newblock Relaxation rates of gene expression kinetics reveal the feedback
  signs of autoregulatory gene networks.
\newblock \emph{J. Chem. Phys.} \textbf{148}, 095102 (2018).

\bibitem[{Olver \emph{et~al.}(2017)Olver, Lozier, Boisvert \&
  Clark}]{olver2010nist}
Olver, F.~W., Lozier, D.~W., Boisvert, R.~F. \& Clark, C.~W.
\newblock N{I}{S}{T} Digital Library of Mathematical Functions  (2017).
\end{thebibliography}

\end{document}